\documentclass[a4paper,UKenglish]{lipics-v2016}
 
\usepackage{microtype}

\usepackage{algorithm}
\usepackage{algorithmicx}
\usepackage{algpseudocode}
\usepackage{enumerate}
\usepackage{amssymb}

\newcommand{\access}{\mbox{\sf access}}
\newcommand{\OPT}{\mbox{\sf OPT}}
\newcommand{\SF}{\mbox{\sf SF}}
\newcommand{\SO}{\mbox{\sf SO}}
\newcommand{\F}{\mbox{\sf F}}
\newcommand{\LF}{\mbox{\sf LF}}
\newcommand{\DF}{\mbox{\sf DF}}
\newcommand{\WS}{\mbox{\sf WS}}
\newcommand{\UB}{\mbox{\sf UB}}

\newcommand{\spn}{\mbox{\sf span}}
\newcommand{\cost}{\mbox{\sf cost}}
\newcommand{\opt}{\mbox{\sf OPT}}
\newcommand{\bset}{{\mathcal{B}}}
\newcommand{\sset}{{\mathcal{S}}}
\newcommand{\mset}{{\mathcal{M}}}
\newcommand{\tset}{{\mathcal T}}

\newcommand{\ceil}[1]{\ensuremath{\left\lceil#1\right\rceil}}

\newcommand{\abs}[1]{\lvert #1\rvert}

\global\long\def\kmin{{\sf UB}^{ {\ell}}}
\global\long\def\kmintree{\ell\text{-DistTree}}
\newcommand{\set}[1]{\left\{ #1 \right\}}

\newcommand{\N}{\ensuremath{\mathbb N}}
\newcommand{\ST}{\ensuremath{\mathcal S}}

\title{Multi-finger binary search trees}

\author[1]{Parinya Chalermsook}
\author[2]{Mayank Goswami}
\author[3]{L\'{a}szl\'{o} Kozma}
\author[4]{Kurt Mehlhorn}
\author[5]{Thatchaphol Saranurak}
\affil[1]{Aalto University, Finland.
  \texttt{parinya.chalermsook@aalto.fi}}
\affil[2]{Queens College, City University of New York.
  \texttt{mayank.goswami@qc.cuny.edu}}
  \affil[3]{TU Eindhoven, Netherlands. \texttt{lkozma@gmail.com}}
  \affil[4]{MPI f\"ur Informatik, Saarbr\"ucken, Germany.
  \texttt{mehlhorn@mpi-inf.mpg.de}}
  \affil[5]{KTH Royal Institute of Technology, Sweden.
  \texttt{thasar@kth.se}}
\authorrunning{P.\ Chalermsook, M.\ Goswami, L.\ Kozma, K.\ Mehlhorn and T.\ Saranurak} 

\Copyright{Parinya Chalermsook, Mayank Goswami, L\'{a}szl\'{o} Kozma, Kurt Mehlhorn, Thatchaphol Saranurak}


\EventEditors{..}
\EventNoEds{1}
\EventLongTitle{arXiv}
\EventShortTitle{arXiv 2018}
\EventAcronym{arXiv}
\EventYear{2018}
\EventDate{}
\EventLocation{}
\EventLogo{}
\SeriesVolume{..}
\ArticleNo{}

\begin{document}

\maketitle

\begin{abstract}
We study multi-finger binary search trees (BSTs), a far-reaching extension of the classical BST model, with connections to the well-studied $k$-server problem. Finger search is a popular technique for speeding up BST operations when a query sequence has \emph{locality of reference}. BSTs with \emph{multiple} fingers can exploit more general regularities in the input. In this paper we consider the cost of serving a sequence of queries in an optimal (offline) BST with $k$ fingers, a powerful benchmark against which other algorithms can be measured.

We show that the $k$-finger optimum can be matched by a standard dynamic BST (having a single root-finger) with an $O(\log{k})$ factor overhead. This result is tight for all $k$, improving the $O(k)$ factor implicit in earlier work. Furthermore, we describe new \emph{online} BSTs that match this bound up to a $(\log{k})^{O(1)}$ factor. Previously only the ``one-finger'' special case was known to hold for an online BST (Iacono, Langerman, 2016; Cole et al., 2000). Splay trees, assuming their conjectured optimality (Sleator and Tarjan, 1983), would have to match our bounds for all $k$.

Our online algorithms are randomized and combine techniques developed for the $k$-server problem with a multiplicative-weights scheme for learning tree metrics. To our knowledge, this is the first time when tools developed for the $k$-server problem are used in BSTs. As an application of our $k$-finger results, we show that BSTs can efficiently serve queries that are \emph{close to some recently accessed item}. This is a (restricted) form of the \emph{unified property} (Iacono, 2001) that was previously not known to hold for any BST algorithm, online or offline.

 \end{abstract}

\section{Introduction}

The binary search tree (BST) is the canonical comparison-based implementation of the dictionary data type for maintaining ordered sets. Dynamic BSTs can be re-arranged after every access via rotations and pointer moves starting from the root. Various ingenious techniques have been developed for dynamically maintaining balanced BSTs, supporting search, insert, delete, and other operations in time $O(\log{n})$, where $n$ is the size of the dictionary (see e.g.~\cite[\S\,6.2.2]{Knuth3}, \cite[\S\,5]{Mehlhorn84}).

In several applications where the access sequence has strong \emph{locality of reference}, the worst-case  bound is too pessimistic (e.g.\ in list merging, adaptive sorting, or in various geometric problems). A classical technique for exploiting locality is \emph{finger search}. In finger search trees, the cost of an access is typically $O(\log{d})$,\footnote{To simplify notation, we let $\log{(x)}$ denote $\log_2{(\max\{2,x\})}$.} where $d$ is the difference in rank between the accessed item and a \emph{finger} ($d$ may be much smaller than $n$). The finger indicates the starting point of the search, and is either given by the user, or (more typically) it points to the previously accessed item. Several special purpose tree-like data structures have been designed to support finger search.\footnote{The initial 1977 design of Guibas et al.~\cite{guibas} was refined and simplified by Brown and Tarjan~\cite{BrownTarjan} and by Huddleston and Mehlhorn~\cite{Huddleston}. Further solutions include~\cite{Tsakalidis, TarjanWyk, Kosaraju, KaplanTarjan}, see also the survey~\cite{BrodalSurvey}. Randomized treaps~\cite{SeidelA96} and skip lists~\cite{SkipList} can also support finger search.}

In 1983, Sleator and Tarjan~\cite{ST85} introduced Splay trees, a particularly simple and elegant ``self-adjusting''  BST algorithm. In 2000, Cole et al.~\cite{finger1,finger2} showed that Splay matches (asymptotically) the efficiency of finger search, called in this context the \emph{dynamic finger} property. This is remarkable, since Splay uses no explicit fingers; every search starts from the root. The result shows the versatility of the BST model, and has been seen as a major (and highly nontrivial) step towards ``dynamic optimality'', the conjecture of Sleator and Tarjan that Splay trees are constant-competitive. 

BSTs can also adapt to other kinds of locality. The \emph{working set} property~\cite{ST85} requires the amortized cost of accessing $x$ to be $O(\log{t})$, where $t$ is the number of distinct items accessed since the last access of $x$. Whereas dynamic finger captures proximity in keyspace, the working set property captures proximity \emph{in time}. In 2001, Iacono~\cite{IaconoUnified} proposed a \emph{unified} property that generalizes both kinds of proximity. Informally, a data structure with the unified property is efficient when accessing an item that is \emph{close to some recently accessed item}. 
{It is not known whether any BST data structure has the unified property.} 

Recently, Iacono and Langerman~\cite{LI16} studied the \emph{lazy finger} property (Bose et al.~\cite{BoseDIL14}), and showed that an online algorithm called Greedy BST\footnote{Greedy BST was discovered by Lucas in 1988~\cite{Luc88} and later independently by Munro~\cite{Mun00}. Demaine et al.~\cite{DHIKP09} transformed it into an online algorithm.} satisfies it. The lazy finger property requires the amortized cost of accessing $x$ to be $O(d)$, where $d$ is the distance (number of edges) from the previously accessed item to $x$ in the best \emph{static reference tree}.  
This property is stronger than the dynamic finger property~\cite{BoseDIL14}, and it is not known to hold for Splay.

In this paper we study a generalization of the lazy finger property; instead of a single finger stationed at the previously accessed item, we allow $k$ fingers to be moved around arbitrarily. An access is performed by moving any of the fingers to the requested item. Cost is proportional to the \emph{total} distance traveled by the fingers. We assume that the fingers move according to an optimal strategy, in an optimally chosen static tree, with a priori knowledge of the entire access sequence. 
The cost of this optimal \emph{offline} execution with $k$ fingers is an intrinsic measure of complexity of a query sequence, and at the same time a benchmark that algorithms in the classical model can attempt to match. 
Parameter $k$ describes the strength of the bound: the case $k=1$ is the lazy finger, at the other extreme, at $k=n$, each item may have its own finger, and all accesses are essentially free. 

Our main result is a family of new \emph{online}\footnote{An \emph{online} BST algorithm can base its decisions only on the current and past accesses. An \emph{offline} algorithm knowns the entire access sequence in advance.} dynamic BST algorithms (in the standard model, where every access starts at the root), matching the $k$-finger optimum on sufficiently long sequences, up to an overhead factor with moderate dependence on $k$ and no dependence on the dictionary size or on the number of accesses in the sequence.

Our online BST combines three distinct techniques: (1) an offline, one-finger BST simulation of a multi-finger execution (the technique is a refinement of an earlier construction~\cite{CombineBST}), (2) online $k$-server algorithms that can simulate the offline optimal multi-finger strategy, and (3) a multiplicative-weights scheme for learning a tree metric in an online fashion.

The fact that ``vanilla'' BSTs can, with a low overhead, simulate a much more powerful computational model further indicates the strength and versatility of the BST model.
As an application, 
we show that our online BST algorithms satisfy a restricted form of the \emph{unified property}; previously no (online or offline) BST was known to satisfy such a property. 

If there is a constant-competitive BST algorithm, then it must match our $k$-finger bounds.  
The two most promising candidates, Splay and Greedy BST (see e.g.\ \cite{in_pursuit}) were only shown (with considerable difficulty) to satisfy variants of the one-finger, i.e.\ lazy finger property. To obtain our online BSTs competitive for other values of $k$, we combine sophisticated tools developed for other online problems, as well as our refinement of a previous (highly nontrivial) construction for simulating multiple fingers. These facts together may hint at the formidable difficulty (more pessimistically: the low likelihood) of attaining dynamic optimality by simple and natural BST algorithms such as Splay or Greedy.

\subparagraph*{BST and finger models. Main results.}
Now, we introduce the formal statements of our results. 
In the dynamic BST model a sequence of keys 
{is} accessed in a binary search tree (BST), and after each access, the tree can be reconfigured via a sequence of rotations and pointer moves starting from the root. (There exist several alternative but essentially equivalent models, see~\cite{Wilber,DHIKP09}.) Denote the space of keys (or elements) by $[n]$. For a sequence $X = (x_1,\ldots, x_m) \in [n]^m$, denote by ${\sf OPT}(X)$ the cost of the optimal offline BST for accessing $X$.\footnote{To avoid technicalities, we only consider \emph{access} (i.e.\ successful search) operations and assume $m \geq n$.}
Arguably the most important question in the BST model is the dynamic optimality conjecture, i.e.\ the existence of an online BST whose cost is $O(\opt(X))$ for every $X$. 

A BST {\em optimality property} is an inequality between ${\sf OPT}(X)$ and some function $f(X)$, that holds in the BST model. (If ${\sf OPT}(X) \leq f(X)$ for all $X$ is a BST optimality property, then every $O(1)$-competitive algorithm must cost at most $O(f(X))$.)

Several natural BST properties have been suggested over the last few decades. For instance, the \emph{static finger} property~\cite{ST85} states $\OPT(X) =O(\SF(X))$, for $\SF(X) = \sum_t \log |x_t - j|$, where 
$j \in [n]$ is a fixed element (finger). The \emph{static optimality} property~\cite{ST85} is $\OPT(X) =O(\SO(X))$, where $\SO(X) = \min_R \sum_{i} d_R(x_i)$. Here $R$ is a \emph{static} BST, and $d_R(x)$ is the depth of $x$ in $R$.   

For the \emph{dynamic finger} property~\cite{ST85}, $\DF(X) =\sum_t \log |x_t - x_{t+1}|$, and for \emph{working set}~\cite{ST85}, $\WS(X) = \sum_t \log \rho_t(x_t)$, where $\rho_t(a)$ is the number of distinct keys accessed between time $t$ and the last time at which $a$ was accessed (all keys assumed accessed at time zero).

In 2001, Iacono~\cite{IaconoUnified} initiated the study of a property that would ``unify'' the latter two notions of efficiency and exhibited a data structure (not a BST) achieving this property. This \emph{unified bound} is defined as $\UB(X) = \sum_{t} \min_{t' < t} \log (|x_t - x_{t'}| + \rho_{t}(x_{t'}))$.  Dynamic finger and working set are in general, not comparable. On the other hand, $\UB(X) \leq \DF(X)$, and $\UB(X) \leq \WS(X)$ clearly hold, justifying the name of the unified bound.

Despite several attempts, the question whether the unified bound is a valid BST property remains unclear; it was shown in~\cite{DerryberryS09} that $\opt(X) = O(\UB(X) +m\log \log n)$, and in~\cite{unified,IaconoUnified} that the unified bound is valid in some other (non-BST) 
models\footnote{Another attempt to study the bounds related to the unified bound was done in~\cite{fresh-finger}.}. 

We show that a unified bound with ``bounded time-window'' holds in the BST model:  

\begin{theorem} 
\label{thm: weak UB} 
For every integer $\ell \geq 1$, every sequence $X$ and some fixed function $\beta(\cdot)$,
\[\opt(X) \leq \beta(\ell) \cdot \UB^{\ell}, \mbox{~~where~~~~} \UB^{\ell} = \sum_t \min_{t'\in [t-\ell,t)} \log \big(|x_t - x_{t'}| + \rho_t(x_{t'})\big) .  \]
\end{theorem}  

Observe that $\UB(X) = \UB^{m}(X) \leq \cdots \leq \UB^1(X) = \DF(X)$.  
Prior to our work it was not known whether the theorem holds when $\ell=2$, i.e.\ no known BST property subsumes this property even when $\ell =2$. 
Thus, Theorem~\ref{thm: weak UB} establishes the first BST property that combines the efficiencies of time- and keyspace-proximity without an additive term.\footnote{The proof of Theorem~\ref{thm: weak UB} implies in fact a stronger, \emph{weighted} form, which we omit for ease of presentation.}

Recently Bose et al.\ \cite{BoseDIL14} introduced the \emph{lazy finger} property, $\LF(X) = \min_R \sum_{i} d_R(x_i,x_{i+1})$. Here distance is measured in a static reference BST $R$, optimally chosen for the entire sequence. The lazy finger bound can be visualized as follows: accesses are performed in the reference tree by moving a unique finger from the previously accessed item to the requested item. The lazy finger property is rather strong: Bose et al.\ show that it implies the dynamic finger and static optimality properties, which in turn imply static finger. 

Our main tool in proving Theorem~\ref{thm: weak UB} is a generalization of the lazy finger property allowing multiple fingers. 
The model is motivated by the famous $k$-server problem.  
For an input sequence $X \in [n]^m$ and a static BST $R$ with nodes associated with the keys in $[n]$, we have $k$ servers located initially at arbitrary nodes in $R$. 
At time $t=1,\ldots, m$, the request $x_t$ arrives, and we move a server of our choice to the node of $R$ that stores $x_t$.
The cost for serving a sequence $X$ is equal to the total movement in $R$ to serve the sequence $X$. 
 
Denote by $\F^{k}_R(X)$ the cost of the optimal (offline) strategy that serves sequence $X$ in $R$ with $k$ servers, minimized over all possible initial server locations. 
Let $\F^{k}(X) = \min_R \F^{k}_R(X)$. 
We call $\F^{k}(X)$ the {\em $k$-finger cost} of $X$.  
We remark that the value of $\F^{k}_R(X)$ is polynomial-time computable for each $R$, $k\in {\mathbb N}$, and $X \in [n]^m$ by dynamic programming. Clearly, $\F^{1}(X) \geq \F^{2}(X) \geq \cdots \geq \F^{n}(X)$ holds for all $X$.  

We first show that one can simulate any $k$-finger strategy in the BST model, in a near-optimal manner. In particular, we prove the following tight result. 

\begin{theorem} 
\label{thm: LF simulation} 
$\opt(X) \leq O(\log k) \cdot \F^{k}(X)$. 
\end{theorem}  

The proof of Theorem~\ref{thm: LF simulation} is a refinement of an earlier argument~\cite{CombineBST}, improving the overhead factor from $O(k)$ to $O(\log{k})$.
The logarithmic dependence on $k$ is, in general, the best possible. To see this, consider a sequence $S$ of length $m$, over $k$ distinct items with average cost $\Omega{(\log{k})}$ (e.g.\ a random sequence from $[k]^m$ does the job). While $\opt{(S)} = \Theta(m\log{k})$, clearly $\F^{k}(X)=O(m)$, as each of the $k$ items can be served with its own private finger.

In the definition of $\F^k(X)$ we assume a \emph{static} reference tree $R$ for the $k$-finger execution. 
The offline BST simulation in the proof of Theorem~\ref{thm: LF simulation} works in fact (with the same overhead) even if $R$ is \emph{dynamic}, i.e.\ if the multi-finger adversary can perform rotations at any of the fingers. In this case, however, the $k$-finger bound is too strong to be useful; already the $k=1$ case captures the dynamic BST optimum. 
Our next result is the online counterpart of Theorem~\ref{thm: LF simulation}. In this case, the restriction that $R$ is static is essential.

\begin{theorem} 
\label{thm: online simulation}
There exists an online randomized BST algorithm whose cost for serving $X \in [n]^m$, is
$O\bigl((\log k)^7\bigr) \cdot \F^{k}(X) + \rho(n)$, for some fixed function $\rho(\cdot)$.
\end{theorem}

The result can be interpreted as follows. On sufficiently long access sequences, there is an online BST algorithm (in fact, a family of them) competitive with the $k$-finger bound, up to an overhead factor with moderate dependence on $k$. 
The randomized algorithm (as is standard in the online setting) assumes an oblivious adversary that does not know in advance the outcomes of the algorithm's random coin-flips. 
The use of randomness seems essential to our approach. We propose as intriguing open questions to find a deterministic online BST with comparable guarantees and to narrow the gap between the online and offline results.

Due to its substantial amount of computation (outside the BST model), our online algorithm is of theoretical interest only. Nonetheless, the connection with the $k$-server problem allows us to ``import'' several techniques to the BST problem; some of these, such as the \emph{double coverage} heuristic for $k$-server~\cite{chrobak} are remarkably simple and may find their way to practical BST algorithms.

The strength of the $k$-finger model lies in the $k$-server abstraction. In order to establish a BST property of the form $\opt(X) \leq \beta(\ell) \cdot O(g(X))$, it is now sufficient to prove $\F^{\ell}(X) \leq \left(\beta(\ell)/\log{\ell}\right) \cdot O(g(X))$.  
In other words, our technique reduces the task of bounding the cost in the BST model to designing $k$-server strategies, which typically admits much cleaner combinatorial arguments. We illustrate this approach by showing that the unified property with a fixed time-window holds in the BST model. 

\begin{theorem} 
\label{thm: strategy for min-dist}
For some fixed functions $\alpha(\cdot),\gamma(\cdot)$, we have:~
 $\F^{\alpha(\ell)}(X) \leq \gamma(\ell) \cdot \UB^{\ell} $.
\end{theorem} 

Theorems~\ref{thm: strategy for min-dist} and~\ref{thm: LF simulation} together imply Theorem~\ref{thm: weak UB}. Moreover, Theorem~\ref{thm: online simulation} implies that the property holds for \emph{online} BST algorithms (we later specify the involved functions).
%
%

The $k$-finger approach can be used to show further BST properties. For example, we connect \emph{decomposability} (refer to \textsection\,4 for definitions) and finger properties by showing that 
even one finger is enough to obtain the \emph{traversal} property in significantly generalized form.

\begin{theorem} \label{thm:decomp}
Let $X$ be a $d$-decomposable sequence. Then $\F^{1}(X) = O(\log d) \cdot |X| $. 
\end{theorem}  

As a corollary, using the recent result by Iacono and Langermann~\cite{LI16}, we resolve an open problem in~\cite{FOCS15}, showing that Greedy costs at most $O(\log d)\cdot |X|$ on every $d$-decomposable sequence, matching the lower bound in~\cite{FOCS15}.\footnote{Independently of our work, Goyal and Gupta~\cite{goyalgupta} showed the same result using a charging argument.} 

In another direction, we connect multiple fingers and generalized monotone sequences. 
In~\cite{FOCS15}, we showed that $\opt(X) \leq |X|\cdot 2^{O(d^2)}$ on every $d$-monotone sequence $X$; a sequence is $d$-monotone if it can be decomposed into $d$ increasing or $d$ decreasing sequences. 
Using the $k$-finger technique, we show the stronger BST property $\opt(X) \leq O(d \log d) \cdot |X|$.

Concerning simple and natural BST algorithms (Splay and Greedy), we give evidence that the strongest results in the literature may still be far from settling the dynamic optimality conjecture. To this end, we describe a class of sequences for which increasing the number of fingers by one  can create an $\Omega(\log n)$ gap. 
More precisely, we show the following: 

\begin{theorem} \label{thm:hierarchyx}
For every integer $k$, there is a sequence $S_k$ such that $\F^{k-1}(S_k) = \Omega(\frac{n}{k} \log (n/k))$ but $\F^{k}(S_k) = O(n)$. 
\end{theorem}

Theorem~\ref{thm:hierarchyx} shows that the multi-finger bounds form a fine-grained hierarchy. For small $k$, our online algorithm (Theorem~\ref{thm: online simulation}) can match these bounds (up to a constant factor). However, any online BST (such as Splay or Greedy) must also match 
the dependence of $O(\log k)$ in the upper bound of $O(\log k) \cdot F^{k}(X)$, in order to be constant-competitive.

\subparagraph*{Techniques. The $k$-server problem.}
The $k$-server problem, introduced by Manasse, McGeoch, and Sleator~\cite{Manasse} in 1988 is a central problem in online algorithms: Is there an online deterministic strategy for serving a sequence of requests by moving $k$ servers around, with a total movement cost at most $k$ times the optimal offline strategy? The question in its original form, for arbitrary metric spaces, remains open. Nonetheless, the problem has inspired a wealth of results and a rich set of techniques, many of which have found applications outside the $k$-server problem. A full survey is out of our scope, we refer instead to some prominent results~\cite{Fiat, Koutsoupias, Seiden,Raghavan, Bartal, Bansal}, and the surveys~\cite[\S\,10, \S\,11]{Borodin},~\cite{kserver_survey}. Most relevantly for us, Chrobak and Larmore~\cite{chrobak} gave in 1991, an intuitive, deterministic, $k$-competitive algorithm for \emph{tree metrics}, and the very recently announced breakthrough of Lee~\cite{Lee}, building on Bubeck et al.\ \cite{Bubeck}, gives an $O\bigl((\log{k})^6\bigr)$-competitive randomized algorithm for arbitrary metrics. 

Our online BST algorithm relies on an online $k$-server in an almost black box fashion (the metric space underlying the $k$-server instance is induced by a static reference BST). Thus, improvements for $k$-server would directly yield improvements in our bounds. Despite the depth and generality of $k$-server (e.g.\ it also models the caching/paging problem), to our knowledge it has previously not been related to the BST problem.\footnote{In his work on a generalized $k$-server problem, Sitters~\cite{Sitters} asks whether the work-function (WF) technique~\cite{Koutsoupias} for $k$-server may have relevance for BSTs. Indeed, we can use WF as an $O(k)$-competitive component of our online BSTs, but for our special case of tree-metrics, the technique of~\cite{chrobak} is much simpler. Whether WF may be used (in different ways) to obtain competitive BSTs remains open.}

It is known that in an arbitrary metric space with at least $k+1$ points, no deterministic online algorithm may have a competitive ratio better than $k$. In the randomized case the lower bound $\Omega(\log{k}/\log{\log{k}})$ holds, see e.g.~\cite{kserver_survey}. (The lower bounds thus apply for a metric induced by a BST, for all $k<n$.) These results imply a remarkable separation between the $k$-server and BST problems. Dynamic optimality would require, by Theorem~\ref{thm: LF simulation}, a BST cost of $O(\log{k})\cdot \F^k$. To match this, an online BST may not implicitly perform a deterministic $k$-server execution, since, in that case its overhead would have to be $\Omega{(k)}$. This indicates that improving Theorem~\ref{thm: online simulation} will likely require tools significantly different from $k$-server, which is surprising, given the similarity of the two formulations. 

Our online BST learns the metric induced by the optimal reference tree using a multiplicative weights update (MWU) scheme. The technique has a rich history, and a recent emergence as a powerful algorithmic tool (we refer to the survey of Arora, Hazan, and Kale~\cite{AroraSurvey}). MWU or closely related techniques have been used previously in data structures (including for BST-related questions), see e.g.\ \cite{Blum, BlumBurch, in_pursuit, Kalai}. Specifically, Iacono~\cite{in_pursuit} obtains, using MWU, an online BST that is constant-competitive on sufficiently long sequences, \emph{if any online BST is constant-competitive}. As we relate online BSTs with an offline strategy, the results are not directly comparable.

\subparagraph*{Further open questions and structure of the paper.}

The main open question raised by our work is whether natural algorithms such as Splay or Greedy match the properties of our new BST algorithms. (This must be the case, if Splay and Greedy are, as conjectured, $O(1)$-competitive). We suggest the following easier questions.
Do Splay or Greedy satisfy the unified bound with a time-window of $2$ steps? Does Splay satisfy the lazy finger or the $2$-monotone bounds? 
Does Greedy satisfy the $2$-finger bound?

Except for Theorems~\ref{thm: LF simulation} and \ref{thm:decomp}, the factors in our results are not known to be tight. Improving them may reveal new insight about the power and limitations of the BST model. 

In \S\,\ref{sec2} we describe our offline BST simulation. In  \S\,\ref{sec3} we describe our new family of online algorithms. In \S\,\ref{sec4} we prove the main applications and further observations.

\section{Offline simulation of multi-finger BSTs (Theorem~\ref{thm: LF simulation})} \label{sec2}
Let $k \in {\mathbb N}$ , let $T$ be a BST on $[n]$, and let $X = (x_1,\dots,x_m) \in [n]^m$ be an access sequence. 
A $k$-\emph{finger strategy} consists of a sequence $\vec{f} \in [k]^m$ where $f_t \in [k]$ specifies the finger that serves access $x_t$. Let $\vec{\ell} \in [n]^k$ be the {\em initial vector}, where $\ell_i \in [n]$ gives the initial location of finger $i$.
The cost of strategy $(\vec{f},\vec{\ell})$ is $\F^{k}_{T,\vec{f}, \vec{\ell}}(X) = \sum_{t=1}^m (1+d_T(x_t, x_{\sigma(f_t,t)}))$ where $\sigma(i,t) = \max\{j < t \mid f_j = i\}$ is the location of finger $i$ before time $t$, and $\sigma(i,1) = \ell_i$.  
Let $\F^{k}_T(X) = \min_{\vec{f}, \vec{\ell}} \F^{k}_{T,\vec{f}, \vec{\ell}} (X)$.  
In other words, for a fixed BST $T$ on keyset $[n]$, $\F^{k}_T(X)$ is the \emph{$k$-server} optimum for serving $X$ in the metric space of the tree $T$. (Note that the tree is unweighted, and the distance $d_T(\cdot,\cdot)$ counts the number of edges between two nodes in $T$.)
We define $\F^{k}(X) = \min_T \F^{k}_{T}(X)$.  
It is clear form the definition that $\F^{1}(X) \geq \F^{2}(X) \geq \cdots \geq \F^{n}(X) = m$ for all $X$.

Observe that we implicitly assume that during every access at most one server moves. In addition, we may assume that if some server is already placed at the requested node, then no movement happens. Algorithms with these two restrictions are called \emph{lazy}. As argued in the $k$-server literature (see e.g.\ \cite{kserver_survey}), non-lazy server movements can always be postponed to a later time, keeping track of the ``virtual'' locations of servers. In other words, every $k$-server algorithm can be simulated by a lazy algorithm, without additional cost. We therefore assume throughout the paper that $k$-server/$k$-finger executions are lazy.

Consider some (lazy) $k$-finger execution $(\vec{f},\vec{\ell})$ in tree $T$, for access sequence $X$. We can view $\vec{f}$ as an explicit sequence of elementary steps $\ST = \ST^k_{T, \vec{f}, \vec{\ell}}$, where in each step we move one of the fingers to its parent or to one of its children in $T$. We further allow $\ST$ to contain rotations at a finger in $T$ (although $k$-finger strategies as described above do not generate rotations). The position of a finger is maintained during a rotation.

We show how $\ST$ can be simulated in a standard dynamic BST. If in $\ST$ a finger visits a node, then the (single) pointer in the BST also visits the corresponding node, therefore all accesses are correctly served in the BST. Every elementary step in $\ST$ is mapped to (amortized) $O(\log{k})$ elementary steps (pointer moves and rotations) in the BST. This immediately implies Theorem~\ref{thm: LF simulation}, since, if we can simulate an arbitrary $k$-finger execution, then indeed we can simulate the optimal $k$-finger execution on the best static tree. Assuming that the intial conditions $T$ and $\vec{\ell}$ are known, the steps of $\ST$ are simulated one-by-one, without any lookahead. Thus, insofar as the $k$-finger execution is \emph{online}, the BST execution is also online (this fact is used in \S\,\ref{sec3}).

Let us describe simulation by a standard BST $T'$ of a $k$-finger execution $\ST$ in a BST $T$. The construction is a refinement of the one given by Demaine et al.\ \cite{CombineBST}, see also~\cite{persistent_tries}. (We improve the overhead factor from $O(k)$ to $O(\log{k})$.)
The main ingredients are: (1) Making sure that each item with a finger
on it in $T$ has depth at most $O(\log k)$ in $T'$. (In~\cite{CombineBST}, each finger may have depth up to $O(k)$ in $T'$.) (2) Implementing a deque data structure within $T'$ 
so that each finger in $T$ can move to any of its neighbors, or perform a rotation, with cost
$O(\log k)$ amortized. (In~\cite{CombineBST}, this cost is $O(1)$ amortized.) 

Given these ingredients, to move a finger $f$ to its neighbor $x$ in $T$, we
can simply access $f$ from the root of $T'$ in $O(\log k)$ steps,
and then move $f$ to $x$ in $T'$ in $O(\log{k})$ amortized steps, with a similar approach for a rotation at $f$. Hence, the overhead factor is $O(\log k)$. We sketch the main technical ideas, postponing the details to Appendix~\ref{sec:simulation lazy fingers}. 

Consider the tree $S$ induced by the current fingers and the paths connecting them in $T$. The tree $S$ consists of finger-nodes and non-finger nodes of degree 3 (both types of nodes are called \emph{pseudo-fingers}), and paths of non-finger nodes of degree 2 connecting pseudo-fingers with each other, called \emph{tendons}. Tendons can be compressed into a BST structure that allows their traversal between the two endpoints in $O(1)$ steps. 

We maintain $S$ as a root-containing subtree of our BST $T'$, called the \emph{hand}. Due to the compression of the tendons, the relevant part of $S$ has size $O(k)$. The description so far, including the terminology, is identical to the one in~\cite[\S\,2]{CombineBST}. Our construction differs in the fact that it maintains the hand, i.e.\ the compressed representation of $S$ as a \emph{balanced} BST. This guarantees the reachability of fingers in $O(\log{k})$ instead of $O(k)$ steps, i.e.\ property (1).

When a finger in $T$ moves or performs a rotation, the designation of some (pseudo)finger, or tendon nodes may change. Such changes can be viewed as the insertion or deletion of items in the tendons. As these operations happen only at certain places within the tendons, they can be implemented efficiently. We implement tendons with the same BST-based \emph{deque} as~\cite{CombineBST}. The construction appears to be folklore, we describe it in Appendix~\ref{deque_proof} for completeness. 

We depart again from~\cite{CombineBST}, as the operation affecting the (pseudo)finger and tendon nodes can trigger a re-balancing of the hand, which may again require $O(\log{k})$ operations to fix, i.e.\ property (2). Any efficient balancing strategy (e.g.\ red-black tree) may be used.

\section{Online simulation of multi-finger BSTs (Theorem~\ref{thm: online simulation})} \label{sec3}

Consider the optimal (offline) $k$-finger execution $\vec{f}$ for access sequence $X \in [n]^m$, with static reference tree $T$ and initial finger-placement $\vec{\ell}$. We wish to simulate it by a dynamic \emph{online} BST.
The construction proceeds in two stages: (1) A simulation of $\vec{f}$ by a sequence $\ST$ of steps that describe finger-movements and rotations-at-fingers, starting from an arbitrary BST $T_0$ and arbitrary finger locations $\vec{\ell}_0$. The sequence $\ST$ is \emph{online}, i.e.\ it is constructed without knowledge of the optimal initial state $T$,$\vec{\ell}$, and it correctly serves the sequence $X$, as its elements are revealed one-by-one. (2) A step-by-step simulation of $\ST$ by a standard BST algorithm using the result of \S\,\ref{sec2}. Since $\ST$ is online, the BST algorithm is also online.

As before, we denote by $\F^k(X) = \F^{k}_{T,\vec{f}, \vec{\ell}}(X)$ the cost of the optimal offline execution. Observe that this is exactly the $k$-server optimum with the tree metric defined by $T$ and initial configuration of servers $\vec{\ell}$. If $T$ and $\vec{\ell}$ were known, we could conclude part (1) by running an arbitrary \emph{online} $k$-server algorithm defined on tree metrics. 

To this end, we mention two online $k$-server algorithms, the deterministic ``double coverage'' algorithm of Chrobak and Larmore~\cite{chrobak} (Algorithm~A) and the very recently announced randomized algorithm of Lee~\cite{Lee, Bubeck} (Algorithm~B). It is known that the cost of Algorithms~A, resp.\ B is at most $k$-times, resp.\ $O((\log{k})^{6})$ times $\F^k$. We only describe Algorithm~A, as it is particularly intuitive. To obtain the claimed result, we need the much more complex Algorithm~B. (By using Algorithm~A we get an overall factor $O(k \log{k})$.) 

During the execution of Algorithm~A, given a current access request $x_t$, call those servers (fingers) \emph{active}, whose path to $x_t$ in $T$ does not contain another server. If several servers are in the same location, one of them is chosen arbitrarily to be active. Algorithm~A serves $x_t$ as follows: as long as there is no server on $x_t$, move all active servers one step closer to $x_t$. Observe that as servers move, some of them may become inactive. Algorithm~A (as described) may need to move multiple servers during one access. It can, however, easily be transformed into a lazy algorithm, as discussed in \S\,\ref{sec2}.

Remains the issue that the optimal initial $T$ and $\vec{\ell}$ are not known. Let $B_1, \dots, B_N$ be instances of an online $k$-server algorithm (in our case Algorithm~B), one for each combination of initial tree $T$ and initial server-placement $\vec{\ell}$. Note that $N = O(4^n \cdot {n^k})$. Let $\mset$ be a ``meta-algorithm'' that simulates all $B_j$'s for $j=1,\dots,N$, competitive on sufficiently long input with the best $B_j$. Algorithm~$\mset$ processes $X$ in epochs of length $M = n \log n$, executing in the $i$-th epoch, for $i=1,\dots,\lceil m/M \rceil$, some $B_{\tau(i)}$ according to a (randomized) choice $\tau(i)$.

Suppose that $\vec{\ell}^*$ and $T^*$ describe the state of $B_{\tau(i)}$ chosen by $\mset$ at the beginning of the $i$-th epoch. To switch to the state $\vec{\ell}^*$, $T^*$, $\mset$ takes $O(n \log{n})$ elementary steps: (1) rotate the current tree to a \emph{balanced} tree using any of the fingers ($O(n)$ steps), (2) move all fingers to their location in $\vec{\ell}^*$ ($k$ times $O(\log{n})$ steps), (3) use an arbitrary finger $f$ to rotate the tree to $T^*$ ($O(n)$ steps), (4) move $f$ back to its location in $\vec{\ell}^*$ ($O(n)$ steps). {Since $M = n \log n$}, the cost of switching can be amortized over the epoch. 

The choice of $B_{\tau(i)}$ for epoch $i$ is done according to the multiplicative-weights (MW) technique~\cite{AroraSurvey}, based on the past performance of the various algorithms. Our \emph{experts} are the online executions $B_1, \dots, B_N$, our $i$-th \emph{event} is the portion of $X$ revealed in the $i$-th epoch, the \emph{loss} of the $j$-th expert for the $i$-th event is the \emph{cost} of $B_j$ in the $i$-th epoch. Let $C_{max}$ denote the maximum possible loss of an expert for an event (we may assume $C_{max} \leq n \cdot M$).

It follows from the standard MW-bounds~\cite[Thm.\ 2.1]{AroraSurvey}, that for an arbitrary $\varepsilon \in (0,1)$, the cost of $\mset$ on $X$ is at most 
{$\min_j (1+\varepsilon)\mathcal{C}_j + \displaystyle\frac{C_{max} \cdot \ln{N}}{\varepsilon}$, where $\mathcal{C}_j$ is the cost of expert $B_j$ for the entire $X$; in particular, $B_j$ may correspond to the optimal offline choice $\vec{\ell}$, $T$, in which case $\mathcal{C}_j=O((\log{k})^6) \cdot \F^k(X)$.}

Thus, for e.g.\ $\varepsilon = 1/2$, we obtain that the cost of $\mset$ on $X$ is at most $O((\log{k})^6) \cdot \F^k(X) + O(n^3 \log^2{n} )$. The output of $\mset$ is an \emph{online} sequence $\ST_\mathcal{M}$ of rotations and finger moves, starting from an arbitrary initial state $T_0$ and $\vec{\ell}_0$. Note that while $\mset$ needs to evaluate the costs and current states for all experts in all epochs (an extraordinary amount of computation), only one of the experts interacts with the tree at any time. Thus, $\ST_\mathcal{M}$ is a standard sequence of steps which can be simulated by a standard BST algorithm according to Theorem~\ref{thm: LF simulation}, at the cost of a further $O(\log{k})$ factor. This concludes the proof of Theorem~\ref{thm: online simulation}.

\section{Applications of the multi-finger property}
\label{sec4}
In this section we show that every BST algorithm that satisfies the $k$-finger property also satisfies the unified bound with fixed time-window (Application 1), is efficient on decomposable sequences (Application 2), and on generalized monotone sequences (Application 3).

\subparagraph*{Application 1. Combined space-time sensitivity (Theorem~\ref{thm: strategy for min-dist}).}

Recall the definition of $\UB^{\ell}$ in Theorem~\ref{thm: weak UB} for a sequence $X=(x_{1},\dots, x_{m}) \in [n]^m$. We connect this quantity with the $k$-finger cost, from which Theorem~\ref{thm: strategy for min-dist} immediately follows.

\begin{theorem}
For every $\ell$, $F^{(\ell!)}(X)=O(\ell!)\cdot\kmin(X)$.\label{thm:LF less than kmin}
\end{theorem}
Since we are only concerned with the case when $\ell$ is constant, we may drop the term $\rho_t(x_{t'})$ in the definition of $\UB^{\ell}$ (whose value is always between $1$ and $\ell$). 

We prove Theorem~\ref{thm:LF less than kmin} via
another bound in which distances are measured in a static reference BST:\quad
$
\displaystyle\kmintree_{T}(X)=\sum_{i=1}^{m}\min_{i-\ell\le j<i}\left\{d_{T}(x_{i},x_{j}) + 1\right\}$. \footnote{We let $x_0$ denote the root of $T$, and distances involving negative indices are defined to be $+\infty$.}

\begin{lemma}
$\min_{T}\kmintree_{T}(X)=O(\kmin(X))$.\label{thm:distree to dist}\end{lemma}
\begin{proof}
By \cite[Thm.\ 4.7]{SeidelA96}, there is a randomized BST $\tilde{T}$ such that the expected distance between
elements $i$ and $j$ is $E[d_{\tilde{T}}(i,j)]=\Theta(\log|i-j|)$. Therefore, 
\begin{align*}
\min_{T}\kmintree_{T}(X)  \le  E[\sum_{i=1}^{m}\min_{i-\ell\le j<i}\{d_{\tilde{T}}(x_{i},x_{j})+1\}]
 =  \sum_{i=1}^{m}E[\min_{i-\ell\le j<i}\{d_{\tilde{T}}(x_{i},x_{j})+1\}]\\
 \quad \le  \sum_{i=1}^{m}\min_{i-\ell\le j<i}\{E[d_{\tilde{T}}(x_{i},x_{j})+1]\}
 =  \sum_{i=1}^{m}\min_{i-\ell\le j<i}\{O(\log|x_{i}-x_{j}|)\}
 =  O(\kmin(X)). \quad \quad \qedhere
\end{align*}
\end{proof}
It is now sufficient to show that $\F^{(\ell!)}_{T}(X)=O(\ell!)\cdot\kmintree_{T}(X)$, for all $X$ and $T$, i.e.\ to describe an $(\ell !)$-finger strategy in $T$ for serving $X$ with the given cost.

At a high level, our strategy is the following: (1) Define a \emph{virtual tree} $\tset(X)$ whose nodes are the \emph{requests} $x_i$ for $i=1,\dots,m$. The virtual tree captures the \emph{proximities} between the requests, with each $x_i$ having as parent the \emph{nearest} request $x_j$ within a fixed time-window before time $i$. Edges in $\tset(X)$ are given as weights the distances between requests in $T$. Note that the virtual tree is not necessarily binary. (2) Define a recursive structural decomposition of the tree $\tset(X)$, with the property that certain blocks of this decomposition contain requests in non-overlapping time-intervals. (3) Describe a multi-finger strategy on $\tset(X)$ for serving the requests, which induces a multi-finger strategy on $T$ with the required cost. (The strategy takes advantage of the decomposition in (2).)

We describe the steps more precisely, deferring some details to Appendix~\ref{weakuniapp}.

\newcommand{\start}{\mathsf{start}}
\newcommand{\myend}{\mathsf{end}}
\newcommand{\nf}{\mathsf{nf}}

\subparagraph*{The virtual tree.}

Given a number $\ell$, $X \in [n]^m$, and a BST $T$ over $[n]$ with root $r$, the virtual tree $\tset = \tset{(\ell,T,X)}$ is a rooted tree with vertex-set $\{(i,x_{i}) \mid i\in[m]\}\cup\{(0,x_{0})\}$, where $x_0 = r$ is the root of $T$ and $(0,x_{0})$ is the root of $\tset$. 
The \emph{parent} of a non-root vertex $(i,x_{i})$ in $\tset$ is 
$(j,x_{j})=\arg\min_{j \in [i-\ell,i)}\{d_{T}(x_{i},x_{j})\}$. In words, $(j,x_j)$ is the request at most $\ell$ steps before $(i,x_i)$, closest to $x_i$ (in $T$).

For each edge $e=((j,x_{j}),(i,x_{i}))$, we define the weight $w_{\tset}(e)=d_{T}(x_{i},x_{j})+1$.
For each subtree $H$ of $\tset$, let $w_{\tset}(H)$
be the total weight of its edges. Observe that $w_{\tset}(\tset)=\kmintree_{T}(X)$.

\subparagraph*{Structure and decomposition of the virtual tree.} We say that a vertex $(i,x_{i})$
is \emph{before} (or \emph{earlier than}) $(j,x_{j})$ if $i<j$, otherwise it is \emph{after} (or \emph{later than}). 
For every subtree $H$ of $\tset$ we denote the earliest vertex in $H$ as $\start(H)$ and the latest vertex in $H$ as $\myend(H)$. The \emph{time-span} of $H$, denoted $\spn(H)$, is $(t_{1},t_{2}]$ where $(t_{1},x_{t_{1}})=\start(H)$ and $(t_{2},x_{t_{2}})=\myend(H)$, and $H$ is \emph{active} at time $t$ if $t \in \spn(H)$.

We describe a procedure to decompose $\tset{(\ell,T,X)}$ into directed paths (for the purpose of analysis), defining the key notions of \emph{$i$-body }and \emph{$i$-core}. The procedure is called on a subtree $H$ of $\tset$, and the top-level call is $\textsf{decompose}(\tset,\ell)$. 
\vspace{-0.05in}

\noindent\rule{\textwidth}{0.4pt}
\vspace{-0.06in}
\noindent\textbf{procedure} $\textsf{decompose}(H,i)$:
\vspace{-0.02in}
\begin{enumerate}
\item If $H$ has no edges, return.
\item Let $C(H)$ be the path from $\start(H)$ to $\myend(H)$. 
\item Call $C(H)$ an $i$-core of $H$, and call $H$ the $i$-body of $C(H)$.
\item For each connected component $H'$ in $H\setminus C(H)$ invoke $\textsf{decompose}(H',i-1)$.
\end{enumerate}
\vspace{-0.05in}
\noindent\rule{\textwidth}{0.4pt}

Observe that $\tset$ itself is an $\ell$-body. 
Each $i$-body $H$ consists of its $i$-core
$C(H)$ and a set of $(i-1)$-bodies 
that are connected
components in $H\setminus C(H)$. For each of those $(i-1)$-bodies
$H'$, we say that $H$ is a \emph{parent} of $H'$, defining a tree-structure over bodies. Observe that the number of ancestor bodies of an $i$-body (excluding itself) is $\ell-i$. We make a sequence of further structural observations about the virtual tree and its decomposition.

\begin{lemma}[\ref{app:struc_decomp}]
\label{lem:struc_decomp}
\begin{enumerate}[(i)]
\item At every time $t$, there are at most
$\ell$ active edges in $\tset{(\ell,T,X)}$. 
\item The $i$-cores of the decomposition, for $1\le i\le \ell$, partition the vertices of $\tset$.
\item Let $H$ be an $i$-body. At any time during the time-span of $H$, among the $(i-1)$-bodies with parent $H$ at most $i-1$ are active.\label{lem:i-1 active children}
\item Let $H$ be an $i$-body. The $(i-1)$-bodies with parent $H$ can be partitioned into $(i-1)$ groups ${\cal H}_{1},\dots,{\cal H}_{i-1}$ such that, for $1\le j\le i-1$
and $H',H''\in{\cal H}_{j}$, the time-spans of $H'$ and $H''$ are disjoint.\label{thm:separate group}
\end{enumerate}
\end{lemma}

\subparagraph*{The strategy for moving fingers.}

For 
two vertices $(i,x_{i})$ and $(j,x_{j})$ in the virtual tree $\tset = \tset{(\ell,T,S)}$, \emph{moving a finger} $f$ from $(i,x_{i})$ to $(j,x_{j})$
means the following: let $P=((i_{1},x_{i_{1}}),\dots,(i_{k},x_{i_{k}}))$
be the unique path from $(i,x_{i}) = (i_{1},x_{i_{1}})$ to $(j,x_{j})=(i_{\ell},x_{i_{\ell}})$ in $\tset$.
For $j=1,\dots,k-1$, we iteratively move a finger $f$ from $x_{i_{j}}$
to $x_{i_{j+1}}$ using $d_{T}(x_{i_{j}},x_{i_{j+1}})$ steps. Hence,
the total number of steps is at most $w_{\tset}(P)$. 

By \emph{serving an access in an $i$-body $H$}, we mean that, for each $(j,x_{j})\in V(H)$, at time $j$ there is a finger move
to $x_{j}$ in $T$. For each $i \le \ell$, let $\nf(i)$ be the number
of fingers used for serving accesses in an $i$-body. We define $\nf(1)=1$
and $\nf(i)=1+(i-1)\cdot\nf(i-1)$, thus, by induction, $\nf(i)\le i!$ for all $i \leq \ell$.

We now describe the strategy for moving fingers. Let $F$ be a set
of fingers where $|F|=\nf(\ell)$. At the beginning all fingers are at $(0,x_0)$. (In the reference tree $T$, all fingers are initially at the root $x_0$.) 
For $1\le j\le m$, we call $\access(\tset,F,(j,x_{j}))$, defined below for an $i$-body $H$, set of fingers $F$, and $u\in V(H)$. 

\noindent\rule{\textwidth}{0.4pt}
\vspace{-0.06in}
\noindent\textbf{procedure} $\textsf{access}(H,F,u)$: 
\vspace{0.05in}

Let $C = C(H)$ be the $i$-core of $H$, with $C=\{u_{1},\dots,u_{k'}\}$, where $u_{k}$ is before $u_{k+1}$ for each $k$. 
For $1\le j\le i-1$, let
${\cal H}_{j}$ be the $j$-th group of the $(i-1)$-bodies with parent $H$ (${\cal H}_{j}$ defined in Lemma~\ref{lem:struc_decomp}(iv)).
The $i$-bodies in ${\cal H}_{j}$ are ordered by their time-span.
That is, suppose ${\cal H}_{j}=\{H'_{1},\dots,H'_{\ell'}\}$. For
each $\ell$, if $\spn(H'_{\ell})=(a_{1},a_{2}]$ and $\spn(H'_{\ell+1})=(b_{1},b_{2}]$,
then $a_{2}\le b_{1}$. Fingers in $F$ are divided into $i$ groups
$F_{1},\dots,F_{i-1},\{f_{i}\}$, where $|F_{j}|=\nf(i-1)$, for $j\le i-1$,
and $f_{i}$ is a single finger.\smallskip
\vspace{-0.06in}
\begin{enumerate}
\item If $u\in C$, then move $f_{i}$ to $u$ from the predecessor node of $u$ in $C$. If $u=\myend(H)$, then move $F$ from $\myend(H)$ to $\start(H)$.

\item 
{Else let $u\in V(H')\setminus V(C)$ where $H'\in{\cal H}_{j}$. If $u = \start(H')$ and $H'$ is the first $(i-1)$-body in ${\cal H}_j$, move $F_j$ from $\start(H)$ to $\start(H')$. Perform  $\access(H',F_{j},u)$. If $u=\myend(H')$ and if $H'$ is the last in ${\cal H}_{j}$ then move $F_{j}$
from $\start(H')$ to $\myend(H)$.
Otherwise, if $u=\myend(H')$ and there is a next $(i-1)$-body $H''$ in ${\cal H}_{j}$, then
move $F_{j}$ from $\start(H')$ to $\start(H'')$.}

\end{enumerate}
\vspace{-0.05in}
\noindent\rule{\textwidth}{0.4pt}

In order to give the reader more intution, we give an alternative description. A $1$-body $H$ consists only of its $1$-core $C(H)$. We use one finger and move it through $C(H)$. For $i > 1$, an $i$-body $H$ decomposes in its $i$-core $C(H)$ and $i-1$ groups ${\cal H}_1$ to ${\cal H}_{i-1}$ of $(i-1)$-bodies. Initially, we have $\nf(i)$ fingers on $\start(H)$. We use one finger to move down the $i$-core. We use a group $F_j$ of $\nf(i-1)$ fingers for the $j$-group ${\cal H}_j$. Let $H_1$, \ldots $H_p$ be the $(i-1)$-cores in ${\cal H}_j$. We first move $F_j$ to $\start(H_1)$. Then we use the strategy recusively to move $F_j$ through $H_1$. Once the group of fingers has reached $\myend(H_1)$, we move them to $\start(H_2)$, and so on. Once the fingers have reached $\myend(H_p)$, we move them back to $\start(H)$. We coordinate (this is not really necessary) the movement of the fingers by the order of the accesses in the access sequence $X$.

From the description of $\access$ it is clear that all
accesses in $\tset$ are served and that $\nf(\ell)$ fingers are sufficient.
It remains to bound the total number of steps all fingers move. For
an $i$-body $H$, let $\cost(H)$ be the total cost of calling
$\access(H,F,u)$ for all $u\in H$. 
Let ${\cal H}$ denote the set of $(i-1)$-bodies with parent $H$. Let $C^{+}(H)$ denote the $i$-core $C(H)$ augmented with the edges connecting $C(H)$ to the $(i-1)$-bodies in ${\cal H}$. 
Then:
\begin{lemma}[\ref{app:finger}]
$\cost(H)\le2\cdot\nf(i)\cdot w_{\tset}(C^{+}(H))+\sum_{H'\in{\cal H}}\cost(H')$.\label{thm:cost_H}\end{lemma} 

By induction, we obtain $\cost(H)\le2\cdot i!\cdot w_{\tset}(H)$. (For $i=1$ we have $H = C(H)$.) 

\noindent Since $\nf(\ell)\le \ell!$, we
have that $\F^{(\ell!)}_{T}(X)\le \F^{\nf(\ell)}_{T}(X) \leq \cost(\tset)$.
By the previous claim we have $\cost(\tset)\le 2\cdot (\ell!)\cdot w_{\tset}(\tset)=2 \cdot (\ell!)\cdot\kmintree_{T}(X)$, concluding the proof.

\subparagraph*{Application 2. Decomposable sequences (Theorem~\ref{thm:decomp}).}

Let $\sigma = (\sigma(1),\ldots, \sigma(n))$ be a permutation. 
For $a,b: 1 \leq a< b \leq n$, we say that $[a,b]$ is a {\em block} of $\sigma$ if $\set{\sigma(a),\ldots, \sigma(b)} = \set{c,\ldots, d}$ for some integer $c,d \in [n]$.
A {\em block partition} of $\sigma$ is a partition of $[n]$ into $k$ blocks $[a_i, b_i]$ such that $(\bigcup_i [a_i, b_i]) \cap \N = [n]$.   
For such a partition, for each $i=1,\ldots, k$, consider a permutation $\sigma_i\in S_{b_i -a_i+1}$ obtained as an order-isomorphic permutation when restricting $\sigma$ on $[a_i, b_i]$.
For each $i$, let $q_i \in [a_i, b_i]$ be a representative element of $i$.  
The permutation $\tilde \sigma \in [k]^k$ that is order-isomorphic to $\set{\sigma(q_1),\ldots, \sigma(q_k)}$ is called a {\em skeleton}  of the block partition. 
We may view $\sigma$ as a {\em deflation} $\tilde \sigma[\sigma_1, \ldots, \sigma_k]$. 

A permutation $\sigma$ is $d$-decomposable if 
$\sigma = (1)$, or  
$\sigma = \tilde \sigma[\sigma_1,\ldots, \sigma_{d'}]$ for some $d' \leq d$ and each permutation $\sigma_i$ is $d$-decomposable (we refer to~\cite{FOCS15} for alternative definitions). Permutations that are $2$-decomposable are called \emph{separable}~\cite{separable}, and this class includes preorder traversal sequences~\cite{ST85} as a special case.

To show Theorem~\ref{thm:decomp}, it is sufficient to define a reference tree $T$ and a one-finger strategy for serving a $d$-decomposable sequence $X$ in $T$ with cost $O(\log{d}) \cdot |X|$. (Appendix~\ref{app:decomp}.)

Combined with the Iacono-Langerman result~\cite{LI16} that Greedy BST has the lazy finger property, we conclude that 
the cost of Greedy on any $d$-decomposable sequence $X$ is at most $O(\log d)\cdot |X|$.  
The result is tight and strengthens our earlier bound~\cite{FOCS15} of $|X| \cdot 2^{O(d^{2})}$.

\subparagraph*{Application 3. Generalized monotone sequences.}

A sequence $X \in [n]^m$ is \emph{$k$-monotone}, if it can be partitioned into $k$ subsequences (not necessarily contiguous), all increasing or all decreasing. This property has been studied in the context of adaptive sorting, and special-purpose structures have been designed to exploit the $k$-monotonicity of input sequences (see e.g.\ \cite{Moffat, Levcopoulos}). Our results show that BSTs can also adapt to such structure.

\begin{theorem}
Let $X$ be a $k$-monotone sequence. Then 
$\F^k(X)=O(k) \cdot |X|$.
\end{theorem}

It follows that $\OPT(X) \leq O(k \log {k}) \cdot |X|$ for $k$-monotone sequences.\footnote{The result holds, in fact, for the more general case, when each $X_i$ is either increasing or decreasing.} The simulation is straightforward. Let $\{X_1, \dots, X_k\}$ be a partitioning of $X$ into increasing sequences (such a partition can be found online).  
Let $T$ be an arbitrary static BST over $[n]$. 
Consider $k$ fingers $f_{1},\dots,f_{k}$, initially all on $1$.
For accessing $x_j \in X_i$, move finger $f_i$ to $x_j$. Observe that over the entire sequence $X$, each finger does only an in-order traversal of $T$, taking $O(n)$ steps. Thus, $\F^k_{T}(X)=O(nk)$.

A lower bound of $\Omega(n \log{k})$ follows from enumerative results: for sufficiently large $n$, the number of $k$-monotone permutations $X \in [n]^n$ is at least $k^{\Omega(n)}$ (implied by e.g.~\cite{regev1981}). Therefore, by a standard information-theoretic argument (see e.g.\ \cite[Thm.\ 4.1]{Blum}), there exists a $k$-monotone permutation $X \in [n]^n$ with $\OPT(X) = \Omega(n \log k)$.

\subparagraph*{Further results.}
\label{sec5}
\label{sec:separation} 

We state our hierarchy result (Theorem~\ref{thm:hierarchyx}), also implying a weak separation between $k$-finger bounds and ``monotone'' bounds. 

\begin{theorem}[Appendix~\ref{sec:proof-hierarchy}]
\label{thm:hierarchy}  
For all $k$ and infinitely many $n$, there is a $k$-monotone sequence $S_k$ of length $n$, such that: 
\begin{itemize} 
\item $F^{k-1}(S_k) = \Omega(\frac{n}{k} \log (n/k))$ 

\item $F^{k}(S_k) = O(n)$ (independent of $k$). 
\end{itemize}
\end{theorem}

In addition, we show a separation between the $k$-finger property and the working set property, showing that for all $k$ and infinitely many $n$, there are sequences $S$ and $S'$ of length $n$, such that $\WS(S) = o(\F^{k}(S))$, and $\F^{k}(S') = o(\WS(S))$. (Appendix~\ref{sec:WS-LF}.)

\newpage

\paragraph*{Acknowledgements}

Parinya Chalermsook is supported by European Research Council (ERC) under the European Union’s Horizon 2020 research and innovation programme (grant agreement No.\ 759557) and by Academy of Finland Research Fellows, under grant No.\ 310415. L\'{a}szl\'{o} Kozma is supperted through ERC consolidator grant No.\ 617951. Thatchaphol Saranurak is supported by European Research Council (ERC) under the European Union's Horizon 2020 research and innovation programme under grant agreement No 715672, and by the Swedish Research Council (Reg.\ No.\ 2015-04659).

We thank Nikhil Bansal and Greg Koumoutsos for insightful discussions.

\newpage
\appendix

\section{Offline BST simulation}
\label{sec:simulation lazy fingers} 

\subsection{BST simulation of a deque}\label{deque_proof}

\begin{lemma} 
	\label{prop:DequeBST}The minimum and maximum element 
	from a BST-based deque can be deleted in $O(1)$ amortized operations.
\end{lemma}
\begin{proof}
The simulation is inspired by the well-known simulation of a queue by two stacks with constant amortized time per operation (\cite[Exercise 3.19]{Mehlhorn-Sanders:Toolbox}). We split the deque at some position (determined by history) and put the two parts into structures that allow us to access the first and the last element of the deque. It is obvious how to simulate the deque operations as long as the sequences are non-empty.  When one of the sequences becomes empty, we split the other sequence at the middle and continue with the two parts.  A simple potential function argument shows that the amortized cost of all deque operations is constant. Let $\ell_1$ and $\ell_2$ be the length of the two sequences, and define the potential $\Phi = |\ell_1 - \ell_2|$. As long as neither of the two sequences are empty, for every insert and delete operation both the cost and the change in potential are $O(1)$. If one sequence becomes empty, we split the remaining sequence into two equal parts. The decrease in potential is equal to the length of the sequence before the splitting (the potential is zero after the split). The cost of splitting is thus covered by the decrease of potential.

The simulation by a BST is easy. We realize both sequences by chains attached to the root. The right chain contains the elements in the second stack with the top element as the right child of the root, the next to top element as the left child of the top element, and so on.

\end{proof}

\subsection{Extended hand} \label{appa2}

To describe the simulation precisely, we borrow terminology from~\cite{CombineBST, persistent_tries}. Let $T$ be a BST with 
a set $F$ of $k$ fingers $f_{1},\dots,f_{k}$. For convenience
we assume the root of $T$ to be one of the fingers. Let $S(T,F)$ be
the Steiner tree with terminals $F$. A \emph{knuckle} is a connected
component of $T$ after removing $S(T,F)$, i.e.\ a hanging subtree of $T$. Let $P(T,F)$ be the union
of fingers and the degree-3 nodes in $S(T,F)$. We call $P(T,F)$
the set of \emph{pseudofingers}. A \emph{tendon} $\tau_{x,y}$ is
the path connecting two pseudofingers $x,y\in P(T,F)$ (excluding $x$ and $y$) such
that there is no other $z\in P(T,F)$ inside. We assume that $x$
is an ancestor of $y$. 

The next terms are new. For each tendon $\tau_{x,y}$, there
are two \emph{half tendons,} $\tau_{x,y}^{<},\tau_{x,y}^{>}$ containing
all elements in $\tau_{x,y}$ which are less than $y$ and greater than
$y$ respectively. Let $H(T,F)=\{\tau_{x,y}^{<},\tau_{x,y}^{>}\mid\tau_{x,y}$
is a tendon$\}$ be the set of all half tendons. 

For each $\tau\in H(T,F)$, we can treat $\tau$ as an interval $[\min(\tau),\max(\tau)]$
where $\min(\tau),\max(\tau)$ are the minimum and maximum elements
in $\tau$ respectively. For each $f\in P(T,F)$, we can treat $f$
as an trivial interval $[f,f]$. 

Let $E(T,F)=P(T,F)\cup H(T,F)$ be the set of intervals defined by
all pseudofingers $P(T,F)$ and half tendons $H(T,F)$. We call $E(T,F)$
an \emph{extended hand}\footnote{In \cite{CombineBST}, the hand is defined only over the pseudofingers.}.
Note that when we treat $P(T,F)\cup H(T,F)$ as a set of elements,
such a set is exactly $S(T,F)$. So $E(T,F)$ can be viewed as a partition
of $S(T,F)$ into pseudofingers and half-tendons. Figure~\ref{ExtendedHand} illustrates these definitions.

\begin{figure}[t]
\begin{center} \includegraphics[width=0.3\textwidth]{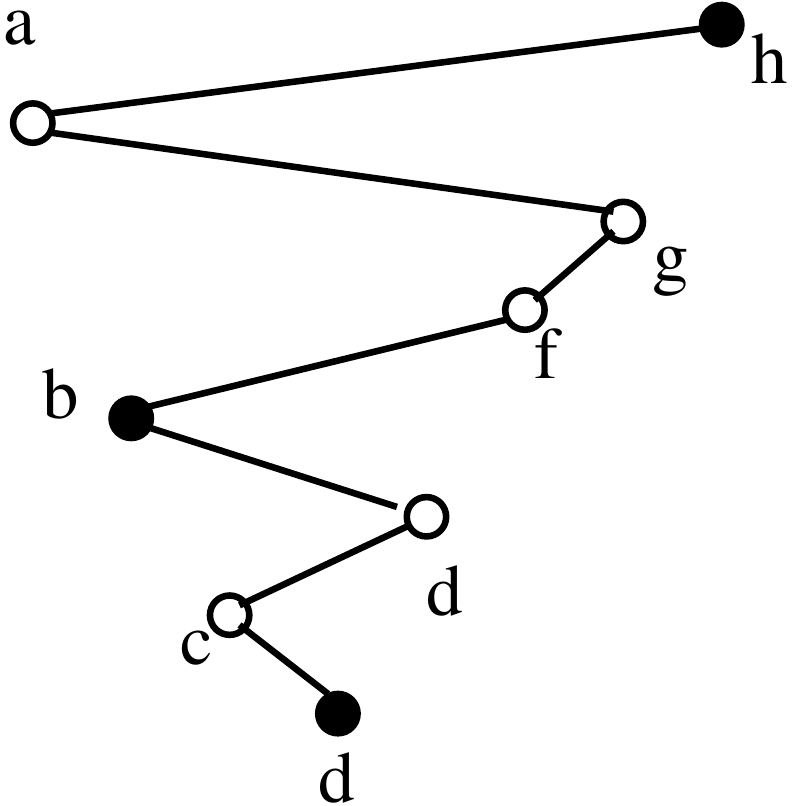}\end{center}
\caption{\label{ExtendedHand} The pseudofingers are $b$, $d$, and $h$. The half-tendons $\tau_{h,b}^<$ and $\tau_{h,b}^>$ are $a$ and $g,f$. The intervals in $E(T,F)$ are $[a,a]$, $[b,b]$, $[c,c]$, $[d,d]$, $[f,g]$, and $[h,h]$. }
\end{figure}

We first state two facts about the extended hand.
\begin{lemma}
	\label{lem:hand size}Given any $T$ and $F$ where $|F|=k$, there
	are $O(k)$ intervals in $E(T,F)$.\end{lemma}
\begin{proof}
	Note that $|P(T,F)|\le2k$ because there are $k$ fingers and there
	can be at most $k$ nodes with degree 3 in $S(T,F)$. Consider the
	graph where pseudofingers are nodes and tendons are edges. That graph
	is a tree. So $|H(T,F)|=O(k)$ as well. \end{proof}
\begin{lemma}
	\label{lem:hand ordered}Given any $T$ and $F$, all the intervals
	in $E(T,F)$ are disjoint.\end{lemma}
\begin{proof}
	Suppose that there are two intervals $\tau,x\in E(T,F)$ that intersect
	each other. One of them, say $\tau$, must be a half tendon. Because
	the intervals of pseudofingers are of length zero and they are distinct,
	they cannot intersect. We write $\tau=\{t_{1},\dots,t_{k}\}$ where
	$t_{1}<\dots<t_{k}$. Assume w.l.o.g.\ that $t_{i}$ is an ancestor
	of $t_{i+1}$ for all $i<k$, and so $t_{k}$ is an ancestor of a
	pseudofingers $f$ where $t_{k}<f$. 
	
	Suppose that $x$ is a pseudofinger and $t_{j}<x<t_{j+1}$ for some
	$j$. Since $t_{j}$ is the first left ancestor of $t_{j+1}$, $x$
	cannot be an ancestor of $t_{j+1}$ in $T$. So $x$ is in the left
	subtree of $t_{j+1}$. But then $t_{j+1}$ is a common ancestor of
	two pseudofingers $x$ and $f$, and $t_{j+1}$ must be a pseudofinger
	which is a contradiction.
	
	Suppose next that $x=\{x_{1},\dots,x_{\ell}\}$ is a half tendon where
	$x_{1}<\dots<x_{\ell}$. We claim that either $[x_{1},x_{\ell}]\subset[t_{j},t_{j+1}]$
	for some $j$ or $[t_{1},t_{k}]\subset[x_{j'},x_{j'+1}]$ for some
	$j'$. Suppose not. Then there exist two indices $j$ and $j'$ where
	$t_{j}<x_{j'}<t_{j+1}<x_{j'+1}$. Again, $x_{j'}$ cannot be an ancestor
	of $t_{j+1}$ in $T$, so $x_{j'}$ is in the left subtree of $t_{j+1}$.
	We know either $x_{j'}$ is the first left ancestor of $x_{j'+1}$
	or $x_{j'+1}$ is the first right ancestor of $x_{j'}$. If $x_{j'}$
	is an ancestor of $x_{j'+1}$, then $x_{j'+1}<t_{j+1}$ which is a
	contradiction. If $x_{j'+1}$ is the first right ancestor of $x_{j'}$,
	then $t_{j+1}$ is not the first right ancestor of $x_{j'}$ and hence
	$x_{j'+1}<t_{j+1}$ which is a contradiction again. Now suppose w.l.o.g.
	$[x_{1},x_{\ell}]\subset[t_{j},t_{j+1}]$. Then there must be another
	pseudofinger $f'$ in the left subtree of $t_{j+1}$, hence $\tau$
	cannot be a half tendon, which is a contradiction. 
\end{proof}

\subsection{The structure of the simulating BST}
\label{appa3}

In this section, we describe the structure of the BST $T'$ that we
maintain given a $k$-finger BST $T$ and the set of
fingers $F$. 

For each half tendon $\tau\in H(T,F)$, let $T'_{\tau}$ be the tree
with $\min(\tau)$ as a root which has $\max(\tau)$ as a right child.
$\max(\tau)$'s left child is a subtree containing the remaining elements
$\tau\setminus\{\min(\tau),\max(\tau)\}$. We implement a \emph{BST simulation of a deque} on this subtree as defined in Appendix~\ref{deque_proof}. By Lemma~\ref{lem:hand ordered},
intervals in $E(T,F)$ are disjoint and hence they are totally ordered.
Since $E(T,F)$ is an ordered set, we can define $T'_{E_{0}}$ to
be a balanced BST such that its elements correspond to elements in
$E(T,F)$. Let $T'_{E}$ be the BST obtained from $T'_{E_{0}}$
by replacing each node $a$ in $T'_{E_{0}}$ that corresponds to a half
tendon $\tau\in H(T,F)$ by $T'_{\tau}$. That is, suppose that the
parent, left child, and right child are $a_{{up}},a_{l}$ and $a_{r}$ respectively.
Then the parent in $T'_{E}$ of the root of $T'_{\tau}$ which is
$\min(\tau)$ is $a_{up}$. The left child in $T'_{E}$ of $\min(\tau)$
is $a_{l}$ and the right child in $T'_{E}$ of $\max(\tau)$ is $a_{r}$. 

The BST $T'$ has $T'_{E}$ as its top part and each knuckle of $T$
hangs from $T'_{E}$ in a determined way. 
\begin{lemma}
	\label{lem:depth pseudofinger}Each element corresponding to pseudofinger
	$f\in P(T,F)$ has depth $O(\log k)$ in $T'_{E}$, and hence in $T'$.\end{lemma}
\begin{proof}
	By Lemma~\ref{lem:hand size}, $|E(T,F)|=O(k)$. So the depth of $T'_{E_{0}}$
	is $O(\log k)$. For each node $a$ corresponding to a pseudofinger $f\in P(T,F)$,
	observe that the depth of $a$ in $T'_{E}$ is at most twice the depth
	of $a$ in $T'_{E_{0}}$ by the construction of $T'_{E}$. 
\end{proof}

\subsection{The cost for simulating the $k$-finger BST}
\label{appa4}

We finally prove the claim on the cost of our BST simulation, which immediately implies Theorem~\ref{thm: LF simulation}. That
is, we prove that whenever one of the fingers in a $k$-finger
BST $T$ moves to its neighbor or rotates, we can update the maintained
BST $T'$ to have the structure as described in the last section with
cost $O(\log k)$.

We state two observations which follow from the structure of our maintained BST $T'$ described in \ref{appa3}. The first observation follows immediately from Lemma~\ref{prop:DequeBST}.
\begin{lemma}
	\label{lem:update tendon}For any half tendon $\tau\in H(T,F)$, we
	can insert or delete the minimum or maximum element in $T'_{\tau}$
	with cost $O(1)$ amortized.\end{lemma}
Next, it is convenient to define a set $A$, called \emph{active set}, as a set
of pseudofingers, the roots of knuckles whose parents are pseudofingers,
and the minimum or maximum of half tendons. 
\begin{lemma}
	\label{lemma:change exended hand}When a finger $f$ in a $k$-finger
	BST $T$ moves to its neighbor or rotates with its parent, the extended
	hand $E(T,F)=P(T,F)\cup H(T,F)$ is changed as follows. 
	\begin{enumerate}
		\item There are at most $O(1)$ half tendons $\tau\in H(T,F)$ whose elements
		are changed. Moreover, for each changed half tendon $\tau$, either
		the minimum or maximum is inserted or deleted. The inserted or deleted
		element $a$ was or will be in the active set $A$.
		\item There are at most $O(1)$ elements added or removed from $P(T,F)$.
		Moreover, the added or removed elements were or will be in the active
		set $A$.
	\end{enumerate}
\end{lemma}

\begin{lemma}
	\label{lem:active set}Let $a\in A$ be an element in the active
	set. We can move $a$ to the root with cost $O(\log k)$ amortized.
	Symmetrically, the cost for updating the root $r$ to become some
	element in the active set is $O(\log k)$ amortized.\end{lemma}
\begin{proof}
	There are two cases. If $a$ is a pseudofinger or a root of a knuckle
	whose parent is pseudofinger, we know that the depth of $a$ was $O(\log k)$
	by Lemma~\ref{lem:depth pseudofinger}. So we can move $a$ to root with cost
	$O(\log k)$. Next, if $a$ is the minimum or maximum of a half tendon
	$\tau$, we know that the depth of the root of the subtree $T'_{\tau}$
	is $O(\log k)$. Moreover, by Lemma~\ref{lem:update tendon}, we can delete
	$a$ from $T'_{\tau}$ (make $a$ a parent of $T'_{\tau}$) with cost
	$O(1)$ amortized. Then we move $a$ to root with cost $O(\log k)$
	worst-case. The total cost is then $O(\log k)$ amortized. The proof
	for the second statement is symmetric.
\end{proof}
\begin{lemma}
	\label{lem:update cost}When a finger $f$ in a $k$-finger BST $T$
	moves to its neighbor or rotates with its parent, the BST $T'$ can
	be updated accordingly with cost $O(\log k)$ amortized.\end{lemma}
\begin{proof}
	According to Lemma~\ref{lemma:change exended hand}, we separate our cost
	analysis into two parts. 
	
	For the fist part, let $a\in A$ be the element to be inserted into
	a half tendon $\tau$. By Lemma~\ref{lem:active set}, we move $a$ to root
	with cost $O(\log k)$ and then insert $a$ as a minimum or maximum
	element in $T'_{\tau}$ with cost $O(\log k)$. Deleting $a$ from some half tendon with cost $O(\log k)$ is symmetric.
	
	For the second part, let $a\in A$ be the element to be inserted into
	a half tendon $\tau$. By Lemma~\ref{lem:active set} again, we move $a$
	to root and move back to the appropriate position in $T'_{E_{0}}$
	with cost $O(\log k)$. We also need rebalance $T'_{E_{0}}$ but this
	also takes cost $O(\log k)$.
\end{proof}

	Finally, we describe the BST simulation of a $k$-finger execution with overhead $O(\log k)$.
	Let $A$ be an arbitrary $k$-finger execution in BST
	$T$. Whenever there is an update in $T$ (i.e.\
	a finger moves to its neighbor or rotates), we update
	the BST $T'$ according to Lemma~\ref{lem:update cost} with cost $O(\log k)$
	amortized. The BST $T'$ is maintained so that its structure is as described
	in Appendix~\ref{appa3}. By Lemma~\ref{lem:depth pseudofinger},
	we can access any finger $f$ of $T$ from the root of $T'$ with
	cost $O(\log k)$. Therefore, the cost of the BST execution is
	at most $O(\log k)$ times the cost of $A$.
This concludes the proof.

\section{Missing proofs for Application 1}
\label{weakuniapp}

\subsection{Proof of Lemma~\ref{lem:struc_decomp}}
\label{app:struc_decomp}

\subparagraph*{Part (i)} \ \\

Suppose that there is some time $t$ when there are $\ell'>\ell$ edges
$\{(j_{k},x_{j_{k}}),(i_{k},s_{i_{k}})\}_{k=1}^{\ell'}$
such that $j_{k}<t\le i_{k}$ for all $k\le \ell'$. Since each node
has a unique parent, $i_{1},\dots,i_{\ell'-1}, i_{\ell'}$ must be distinct and hence $\max_{1 \le k \le \ell'} i_{k} \ge t + \ell' - 1 \ge t+\ell$. Thus  $\max_{1 \le k \le \ell'} j_{k} \ge t$, a contradiction.

\subparagraph*{Part (ii)} \ \\

By construction, the cores are edge-disjoint, and every vertex belongs to some core (the recurrence ends on singleton vertices only). It remains to show that when $\textsf{decompose}(H,0)$ is called during the execution of $\textsf{decompose}(\tset,\ell)$,
$H$ has no edges, i.e.\ there is no $i$-core or $i$-body with $i \leq 0$.

To see this, define the sequence of graphs $H_{0},\dots,H_{\ell}$ where $H_{\ell}=\tset{(\ell,T,X)}$,
$H_{i-1}$ is a connected component of $H_{i}\setminus C(H_{i})$,
and $H_{0}=H$. Recall that $\spn(K)$ denotes the time-span of $K$. By definition of $C(H_{i})$, we have $\spn(H_{i-1}) \subseteq \spn(H_{i})$. 

Suppose for contradiction that $H_{0}$ has an edge. Denote $\spn(H_0) = (t_{1},t_{2}]$, where $t_{1}<t_{2}$. For all $0\le i\le \ell$, it holds that $\spn(H_{i}) \supseteq (t_{1},t_{2}]$.
Let $t\in(t_{1},t_{2}]$. We have that $C(H_{i})$ contains an edge $((a_{i},x_{a_{i}}),(b_{i},x_{b_{i}}))$
where $a_{i}<t\le b_{i}$ for all $0\le i\le \ell$. Since $C(H_{i})$
are edge-disjoint, this contradicts part (i).

\subparagraph*{Part (iii)} \ \\

Suppose there are $i$ active $(i-1)$-bodies $H'_{1},\dots,H'_{i}$
of $H$ at time $t$. Since $H$ is an $i$-body, there are $\ell-i$
ancestors $A_{1},\dots,A_{\ell-i}$ of $H$.
For each of the cores $C\in\{C(H'_{1}),\dots,C(H'_{i}),C(H),C(A_{1}),\dots,C(A_{\ell-i}$)\}
which is a set of size $\ell+1$, there is an edge $(a,s_{a}),(b,s_{b})$
where $a<t\le b$. This contradicts part (i).

\subparagraph*{Part (iv)} \ \\

We construct the decomposition greedily. Consider the $(i-1)$ bodies $H'$ ordered by $\start(H')$ and put $H'$ into the group ${\cal H}_j$ for the smallest index $j$ such that the time-span of $H'$ is disjoint from the time-spans of all members of the group. Assume that this process opens up $i' > i-1$ groups. Then
there are $(i-1)$-bodies $H'_1$ to $H'_{i'}$ (one per group) such that the time-span of the $i$-body $H$ intersects the time-spans of $H'_1$ to $H'_{i'}$, contradicting part (iii).

\subsection{Proof of Lemma~\ref{thm:cost_H}}
\label{app:finger}
We analyze the total cost of calling $\access(H,F,u)$
for all $u\in V(H)$. The total cost due to recursive calls in Step 2 is accounted by the term $\sum_{H'\in {\cal H}}\cost({H'})$. The remaining operations amount to 
moving $\nf(i)$ fingers from $\start(H)$ to $\myend(H)$ and back, along the $i$-core $C(H)$. The cost of this is
exactly $2 \cdot \nf(i)\cdot w_{\tset}(C(H))$. In addition we need to traverse, using $\nf(i-1)$ fingers, the edges connecting $C(H)$ to $\start(H')$, twice for all $H'\in {\cal H}$. The total cost thus becomes at most $2 \cdot \nf(i)\cdot w_{\tset}(C^{+}(H)) + \sum_{H'\in {\cal H}}\cost({H'}) $.

We argue now by induction that for an $i$-body $H$, we have $\cost(H)\le2\cdot i!\cdot w_{\tset}(H)$.
For $i=1$, $H=C(H)=C^{+}(H)$. Thus, by the inductive step: 
\[ \cost(H) \le 2 \cdot \nf(1)\cdot w_{\tset}(C^{+}(H))
 \le 2 \cdot w_{\tset}(H).\]
For the general inductive step: 
\begin{align*}
\cost(H) & \le 2 \cdot \nf(i)\cdot w_{\tset}(C^{+}(H))+\sum_{H'\in{\cal H}} \cost({H'}) \\
& \le 2\cdot i!\cdot w_{\tset}(C^{+}(H))+\sum_{H' \in {\cal H}}2\cdot(i-1)!\cdot w_{\tset}(H')\\
 & \le 2\cdot i!\cdot \Bigl(w_{\tset}(C^{+}(H))+\sum_{H' \in {\cal H}} w_{\tset}(H') \Bigr)\\
 & = 2\cdot i!\cdot w_{\tset}(H). 
\end{align*}

\section{Decomposable Sequences}
\label{app:decomp}

\begin{lemma} Let $X = (x_1,\ldots, x_n)$ be a $k$-decomposable permutation of length $n$. Then $\F^1(X) \le 4(\abs{X} - 1) \ceil{\log k}$.
\end{lemma}
\vspace{-0.1in}
\begin{proof} 
It is sufficient to define a reference tree $T$ for which $\F^1_T(X)$ achieves such bound. 
We remark that the tree will have auxiliary elements.  
  We construct $T$ recursively. 
If $X$ has length one, then $T$ has a single node and this node is labeled by the key in $X$. 
Clearly, $\F^1_T(X) = 0$.

Otherwise, let $X = \tilde X[X_1,\ldots,X_j]$ with $j \in [k]$ be the outermost partition of $X$. 
	Denote by $T_i$ the tree for $X_i$ that has been inductively constructed.  
Let $T_0$ be a BST of depth at most $\ceil{\log j}$ and with $j$ leaves.
Identify the $i$-th leaf with the root of $T_i$ and assign keys to the internal nodes of $T_0$ such that the resulting tree is a valid BST. 
Let $r_i$ be the root of $T_i$, $0 \le i \le j$ and let $r = r_0$ be the root of $T$. 
Then

\vspace{-0.15in}
\begin{align*}
d_T(r,x_1) & \le \ceil{\log k} + d_{T_1}(r_1,x_1) \\
d_T(r,x_n) & \le  \ceil{\log k} + d_{T_j}(r_j,x_n) \\
d_T(x_{t-1},x_t) & \le  \begin{cases} d_{T_\ell}(x_{t-1},x_t)    & \text{if $x_{t-1},x_t \in X_\ell$}\\
                    2\ceil{\log k} + d_{T_\ell}(r_\ell,x_{t-1}) + d_{T_{\ell + 1}}(r_{\ell+1},x_t)  &\text{if $x_{t-1} \in X_\ell$ and $x_t \in X_{\ell + 1}$,}\end{cases}
\end{align*}
and hence 
\begin{align*}
\F^1_T(X) & = d_T(r,x_0) + \sum_{t \ge 2} d_T(x_{t-1},x_t) + d_T(x_n,r) \\
	&\le 2j \ceil{\log k}+ \sum_{1 \le \ell \le j} \F^1_{T_{\ell}}(X_{\ell}) \quad
\le \quad 2j  \ceil{\log k} + \sum_{1 \le \ell \le j} 4(\abs{X_\ell} - 1) \ceil{\log k}\\
	&\le (2j - 4j + 4\sum_{1 \le \ell \le j} \abs{X_\ell}) \ceil{\log k}\ \  \le \ \ 4(\abs{X} - 1)\ceil{\log k},
\end{align*}
where the last inequality uses $j \ge 2$. 
\end{proof}

\section{Finger bounds with auxiliary elements }

Recall that $\F(X)$ is defined as the minimum over all BSTs $T$ on $[n]$ of $\F_T(X)$. 
It is convenient to define a slightly stronger finger bound that also allows {\em auxiliary elements}. 
Define $\widehat{\F}(X)$ as the minimum over all BSTs $T$ that contain the keys $[n]$ (but the size of $T$ can be much larger than $n$).
We define $\widehat{\F}^k(X)$ as the $k$-finger bound when the tree is allowed to have auxiliary elements.  
We argue that the two definitions are equivalent. 
  
\begin{theorem}
For any integer $k$, 
	$\F^k(X)=\Theta(\widehat{\F}^k(X))$ for all $X$.\end{theorem}
\begin{proof}
	It is clear that $\widehat{\F}^k(X)\le \F^k(X)$. 
We only need to show the converse.  

	Let $T$ be a BST (with auxiliary elements) such that $\F^k_{T}(X)=\widehat{\F}^k(X)$. Denote by $\vec{f}$ the optimal finger strategy on $T$. 
	Let $[n]\cup Q$ be the elements of $T$ where $Q$ is the set of
	auxiliary elements in $T$. 
For each $a\in[n]\cup Q$, let $d_{T}(a)$
	be the depth of key $a$ in $T$, and let $w(i)=4^{-d_T(i)}$. 
For any
	two elements $i$ and $j$ and set $Y \subseteq [n] \cup Q$, let $w_{Y}[i:j]$ be the
{sum $\sum_{k \in Y\cap [i,j]} w(k)$ of the weights of the elements in $Y$ between $i$ and $j$ (inclusive).}
	For any $i,j\in[n]\cup Q$ such that $i\leq j$, we have  
	\[
	\log\frac{w_{[n]\cup Q}[i:j]}{\min(w(i),w(j))}=O(d_{T}(i,j)),
	\]
	where $d_{T}(i,j)$ is the distance from $i$ to $j$ in $T$. So,
	this same bound also holds when considering only keys in $[n]$. That is, for $i,j \in [n]$, we have  \[
	\log\frac{w_{[n]}[i:j]}{\min(w(i),w(j))}=O(d_{T}(i,j)).
	\]

Given the weight of $\{w(a)\}_{a\in[n]}$, the BST $T'$ (without auxiliary elements) is constructed by invoking Lemma~\ref{thm:tree-from-weight}. 
We bound the term $\F^k_{T'}(X)$ (using strategy $\vec{f}$) by 	
\[
	O( \sum_{t} d_{T'}(x_{\sigma(f_t,t)}, x_t)) = O(\sum_{t=1}^{m-1}\lg\frac{w_{[n]}[x_{t}:x_{\sigma(f_t,t)}]}{\min(w(x_{i}),w(x_{\sigma(f_t,t)}))})=O(\sum_{t=1}^{m-1}d_{T}(x_{\sigma(f_t,t)},x_{t}))=O(\F^k_{T}(X))
	\]
	where $X=(x_{1},\dots,x_{m})$. Therefore, $\F^k(X)\le \F^k_{T'}(X)=O(\F^k_{T}(X))=O(\widehat{\F}^k(X))$.
\end{proof}

\begin{lemma}
Given a weight function $w(\cdot)$, and $W = \sum_{i\in[n]}{w(i)}$, there is a deterministic construction of a BST $T_w$ such that the depth of every key $i \in [n]$ is $d_{T_w}(i) = O(\log \frac{W}{w(i)})$. \label{thm:tree-from-weight} 
\end{lemma}  

\begin{proof}
Let $w_1$, \ldots $w_n$ be a sequence of weights. We show how to construct a tree in which the depth of element $\ell$ is $O(\log w[1:\ell]/\min(w_1,w_\ell))$. 

For $i \ge 1$, let $j_i$ be minimal such that $w[1:j_i] \ge 2^i w_1$. Then $w[1:j_i - 1] < 2^i w_1$ and 
$w[j_{i-1} + 1 : j_i] \le 2^{i-1} w_1 + w_{j_i}$. 

Let $T_i$ be the following tree. The right child of the root is the element $j_i$. The left subtree is a tree in which element $\ell$ has depth $O(\log 2^{i-1} w_1/w_\ell)$. 

The entire tree has $w_1$ in the root and then a long right spine. The trees $T_i$ hang off the spine to the left. In this way the depth of the root of $T_i$ is $O(i)$. 

Consider now an element $\ell$ in $T_i$. Assume first that $\ell \not= j_i$. The depth is
\[ O\left(i + \log \frac{2^{i-1} w_1}{w_\ell}\right) = O\left(i + \log \frac{2^{i-1} w_1}{\min(w_1,w_\ell)}\right)= 
O\left(\log \frac{2^{i-1} w_1}{\min(w_1,w_\ell)}\right)= O\left(\frac{w[1:\ell]}{\min(w_1,w_\ell)}\right).\]

For $\ell = j_i$, the depth is 
\[ O\left(i\right) = O\left( \log \frac{2^i w_1}{w_1} \right) = O\left(\log \frac{w[1:j_i]}{\min(w_1,w_{j_i})}\right) .\]

\end{proof}

\section{Proof of Theorem~\ref{thm:hierarchyx}}
\label{sec:proof-hierarchy}
Let $n$ be an integer multiple of $k$ and $\ell = n/k$.  
Consider the tilted $k$-by-$\ell$ grid $S_k$. 
More precisely, the access sequence is defined as: $1$, $\ell + 1$, \ldots, $\ell \cdot(k-1) + 1$, $2$, $\ell+2$, \ldots, $(k-1)\ell +2$,\ldots, $(k-1)\ell +\ell$. We denote the elements of $S_k$ as $s_i$, for $i=1,\dots,n$.
To see the geometry of this sequence, one may view it as a partitioning of the keys $[n]$ into ``blocks'' $\bset_i: i =1,\ldots, k$ where $\bset_i$ contains the keys in $\{\ell (i-1) +1, \ell (i-1) +2, \ldots, \ell i\}$, so we have $|\bset_i | = \ell$ and $\bigcup_{i=1}^k \bset_i = [n]$.  
The sequence $S_k$ consists of an interleaving of an increasing traversal of each block. 

\begin{lemma}
	$\F^k(S_k) = O(n)$. 
\end{lemma} 
\begin{proof}
The main idea is to use each finger to serve only the keys inside blocks and to use a separate finger for each block. (recall that there are $k$ blocks and $k$ fingers.)  
We create a reference tree $T$ and argue that $\F^k_T(S_k) = O(n)$.
Let $T_0$ be a BST of height $O(\log k)$ and with $k$ leaves. 
Each leaf of $T_0$ corresponds to the keys $\set{\ell \cdot (i-1) +\frac 1 2}_{i=1}^k$. 
The non-leafs of $T_0$ are assigned arbitrary fractional keys that are consistent with the BST properties.
	For each $i \in [k]$, path $P_i$ is defined as a BST with key $\ell \cdot (i-1) +1$ (i.e. the smallest key in block $\bset_i$) at the root, where for each $j = 0, \ldots, (\ell-1)$, the key $\ell (i-1)+j$ has $\ell (i-1)+ (j+1)$ as its only (right) child.  
The final tree $T$ is obtained by hanging each path $P_i$ as a left subtree of a leaf $\ell \cdot (i-1) +\frac 1 2$.  
The $k$-finger strategy is simple: The $i$-th finger only takes care of the elements in block $\bset_i$. 
The cost for the first access in block $\bset_i$ is $O(\log k)$, and afterwards, the cost is only $O(1)$ per access. 
So the total access cost is $O(\frac{n}{k} \log k + n) = O(n)$. 
\end{proof} 

The rest of this section is devoted to proving the following: 
 
\begin{theorem} 
	$\F^{k-1}(S_k) = \Omega(\frac{n}{k} \log (n/k))$
\end{theorem}
Let $T$ be an arbitrary reference tree. We argue that $\F^{k-1}_T(S_k) = \Omega(\frac{n}{k} \log (n/k))$. 

A {\em finger configuration} $\vec{f} = (f(1),\ldots, f(k-1)) \in [n]^{k-1}$ specifies to which keys the fingers are currently pointing. 
Any finger strategy can be described by a sequence $\vec{f}_1,\ldots, \vec{f}_n$, where $\vec{f}_t$ is the configuration after element $s_t$ is accessed.  
As before, we assume w.l.o.g.\ the following lazy update strategy:
\begin{lemma}
\label{prop:k-server}  
For each time $t$, the configurations $\vec{f}_t$ and $\vec{f}_{t+1}$ differ at exactly one position. In other words, we only move the finger that is used to access $s_{t+1}$.  
\end{lemma} 

We view the input sequence $S_k$ as having $\ell$ phases: The first phase contains the subsequence $1, \ell+1,\ldots, \ell(k-1)+1$, and so on. Each phase is a subsequence of length $k$ that accesses keys starting in block $
bset_1$ and so on, until the block $\bset_k$. 

\begin{lemma} 
For each phase $p \in \{1,\ldots, \ell \}$, there is a time 
{$t \in [(p-1)k + 1, p \cdot k]$} such that $s_t$ is accessed by finger $j$ such that $f_{t-1}(j) $ and $f_t(j)$ are in different blocks, and $f_{t-1}(j) < f_t(j)$.  
That is, this finger moves to the block $\bset_{b}$, $b = t \mod k$, from some block $\bset_{b'}$, where $b' < b$, in order to serve $s_t$.  
\end{lemma} 

\begin{proof} 
	Consider the accesses in blocks $\bset_1$, \ldots, $\bset_k$ in order. After the access in $\bset_1$, we have a finger in $\bset_1$ and hence at most $k - 2$ fingers in blocks $\bset_2$, \ldots, $\bset_k$. If the access to $\bset_2$ is served by a finger being in block $\bset_1$ before the acces, we are done. Otherwise, it is server by a finger being in blocks $\bset_{\ge 2}$ before the access. Then we have two fingers in blocks $\bset_{\le 2}$ after the access and at most $k-3$ fingers in blocks $\bset_{\ge 3}$. Continuing in this way, we will find the desired access.
\end{proof} 

For each phase $p \in [\ell]$, let $t_p$ denote the time for which such a finger moves across the blocks from left to right; if they move more than once, we choose $t_p$ arbitrarily.  
Let $J = \{t_p\}_{p=1}^{\ell}$.   
For each finger $j \in [k-1]$, each block $i \in [k]$ and block $i' \in [k]: i < i'$, 
let $J(j,i,i')$ be the set containing the time $t$ for which finger $f(j)$ is moved from block $\bset_i$ to block $\bset_{i'}$ to access $s_t$. Let $c(j,i,i') = |J(j,i,i')|$. 
Notice that $\sum_{j,i,i'} c(j,i,i') = \frac{n}{k} = \ell$, due to the lemma. 
Let $P(j,i,i')$ denote the phases $p$ for which $t_p \in J(j,i,i')$.

\begin{lemma} 
$\sum_{j, i, i': c(j,i,i') \geq 16} c(j,i,i') \geq n/2k$ if $n = \Omega(k^4)$.  
\end{lemma} 
\begin{proof}
There are only at most $k^3$ triples $(j,i,i')$, so the terms for which $c(j,i,i') < 16$ contribute to the sum at most $16k^3$. 
This means that the sum of the remaining is at least $n/k - 16k^3 \geq n/2k$ if $n$ satisfies $n= \Omega(k^4)$.  
\end{proof}

From now on, we consider the sets $J'$ and $J'(j,i,i')$ that only concern those $c(j,i,i')$ with $c(j,i,i') \geq 16$ instead.  

\begin{lemma}
\label{lem:key-lemma}
There is a constant $\eta >0$ such that the total access cost during the phases $P(j,i,i')$ is at least $\eta c(j,i,i') \log c(j,i,i')$.  
\end{lemma} 

Once we have this lemma, everything is done. 
Since the function $g(x) = x \log x$ is convex, we apply Jensen's inequality to obtain: 
\[\frac{1}{|J'|} \sum_{j,i,i'} \eta c(j,i,i')\log c(j,i,i') \geq \eta \cdot \frac{n}{2k|J'|} \cdot \log (n/2k|J'|). \]
Note that the left side is the term ${\mathbb E}[g(x)]$, while the right side is $g({\mathbb E}(x))$. Therefore, the total access cost is at least $\frac{\eta n}{8k} \log (n/2k)$.  
We now prove the lemma. 

\begin{proof}[Proof of Lemma~\ref{lem:key-lemma}] 
We recall that, in the phases $P(j,i,i')$, the finger-$j$ moves from block $\bset_i$ to $\bset_{i'}$ to serve the request at corresponding time.  
For simplicity of notation, we use $\tilde J$ and $C$ to denote $J(j,i,i')$ and $c(j,i,i')$ respectively. 
Also, we use $\tilde f$ to denote the finger-$j$.  
For each $t \in \tilde J$, let $a_t \in \bset_i$ be the key for which the finger $\tilde f$ moves from $a_t$ to $s_t$ when accessing $s_t \in \bset_{i'}$. 
Let $\tilde J = \{t_1,\ldots, t_{C}\}$ such that $a_{t_1} < a_{t_2} <  \ldots < a_{t_{C}}$. 
Let $R$ be the lowest common ancestor in $T$ of keys in $[a_{t_{\lfloor C/2 \rfloor } +1}, a_{t_{C}}]$.

\begin{lemma} 
\label{lem:lastclaim}
For each $r \in \{1,\ldots, \lfloor C/2 \rfloor \}$, the access cost of $s_{t_r}$  and $s_{t_{C -r}}$ is together at least 
$\min \{d_T({R}, s_{t_r}), d_T(R, s_{t_{C -r}})\}$. 
\end{lemma} 

\begin{proof} 
 
Let $u_r$ be the lowest common ancestor between $a_{t_r}$ and $s_{t_r}$. 
Then the cost of accessing $s_{t_r}$ is at least $d_T(u_r,s_{t_r})$. If $s_{t_r}$ is in the subtree rooted at $R$, then $u_r$ must be an ancestor of $R$ (because $a_{t_r} < a_{t_{\lfloor C/2 \rfloor }} < a_{t_{C}} < s_{t_r}$) and hence $d_T(u_r, s_{t_r}) \geq d_T(R,s_{t_r})$. Thus the cost it at least $d_T(R,s_{t_r})$. 
 Otherwise, we know that $s_{t_r}$ is outside of the subtree rooted at $R$, and so is $s_{t_{C- r}}$. On the other hand, $a_{t_{C-r}}$ is in such subtree, so moving the finger from $a_{t_{C -r}}$ to $s_{t_{C - r}}$ must touch ${R}$, therefore costing at least $d_T({R}, s_{t_{C-r}})$. 
\end{proof} 

Lemma~\ref{lem:lastclaim} implies that, for each $r=1,\ldots, \lfloor C/2 \rfloor $, we pay the distance between some element $v_r \in \set{s_{t_r}, s_{t_{C-r}}}$ to $R$. 
The total such costs would be $\sum_{r} d_T(R, v_r)$. 
Applying the fact that (i) $v_r$'s are different and (ii) there are at most $3^d$ vertices at distance $d$ from a vertex $R$, we conclude that this sum is at least $\sum_r d_T(R, v_r) \geq \Omega(C \log C)$.   
\end{proof}

\section{Working set and $k$-finger bounds are incomparable}
We show the following theorem. 

\label{sec:WS-LF}

\begin{theorem} \begin{enumerate}[(1)]
\item There exists a sequence $S$ such that $\WS(S) = o(\F^{k}(S))$, and
\item There exists a sequence $S^{'}$ such that $\F^{k}(S^{'}) = o(\WS(S^{'}))$.
\end{enumerate}
\end{theorem}

The sequence $S^{'}$ above is straightforward: For $k=1$, just consider the sequential access $1,\dots,n$ repeated $m/n$ times. For $m$ large enough, the working set bound is $\Omega(m \log n)$. However, if we start with the finger on the root of the tree which is just a path, then the lazy finger bound is $O(m)$. The $k$-finger bound is always less than lazy finger bound, so this sequence works for the second part of the theorem.

The existence of the sequence $S$ is slightly more involved (the special case for $k=1$ was proved in~\cite{BoseDIL14}), and is guaranteed by the following theorem, the proof of which comprises the remainder of this section.

\begin{theorem}\label{wslessklf}
For all $k = O(n^{1/2 - \epsilon})$, there exists a sequence $S$ of length $m$ such that $\WS(S) = O(m \log k)$ whereas $F^{k}(S) = \Omega(m \log (n/k))$. 
\end{theorem}

We construct a random sequence $S$ and show that while $\WS(S) = O(m \log k)$ with probability one, the probability that there exists a tree $\tset$ such that $\F^{k}_{T}(S) \leq c m\log_{3} (n/k)$ is less than $1/2$ for some constant $c<1$. This implies the existence of a sequence $S$ such that for all trees $\tset$, $\F^{k}_{T}(S) = \Omega(m \log (n/k))$.

The sequence is as follows. We have $Y$ phases. In each phase we select $2k$ elements $R_{i} = \{r^{i}_{j}\}_{j=1}^{2k}$ uniformly at random from $[n]$. We order them arbitrarily in a sequence $S_{i}$, and access $[S_{i}]^{X/2k}$ (access $S_{i}$ $X/2k$ times). The final sequence $S$ is a concatenation of the sequences $[S_{i}]^{X/2k}$ for $1 \leq i \leq Y$. Each phase has $X$ accesses, for a total of $m= XY$ accesses overall. We will choose $X$ and $Y$ appropriately later.

\subparagraph*{\bf Working set bound.} One easily observes that $\WS(S) =O(Y(2k \log n + (X-2k) \log (2k)))$, because after the first $2k$ accesses in a phase, the working set is always of size $2k$. We choose $X$ such that the second term dominates the first, say $X \geq 5 k \frac{\log n}{\log 2k}$. We then have that the working set bound is $O(XY \log k) = O(m \log k)$, with probability one.

\subparagraph*{\bf $k$-finger bound.} Fix a BST $\tset$. We classify the selection of the set $R_{i}$ as being $d$-good for $\tset$ if there exists a pair $r^{i}_{j}, r^{i}_{\ell} \in R_{i}$ such that their distance in $\tset$ is less than $d$. The following lemma bounds the probability of a random selection being $d$-good for $\tset$. 

\begin{lemma} Let $\tset$ be any BST. 
The probability that $R_{i}$ is $d$-good for $\tset$ is at most $8k^{2}3^{d}/n$.
\end{lemma}

\begin{proof} We may assume $8 k^2 3^d/n < 1$ as the claim is void otherwise. We compute the probability that a selection $R_{i}$ is not $d$-good first. This happens if and only if the balls of radius $d$ around every element $r^{i}_{j}$ are disjoint. The volume of such a ball is at most $3^d$, so we can bound this probability as 

\begin{eqnarray}
P[R_{i}\text{ is not d-good for }\tset] &=& \displaystyle\Pi_{i=1}^{2k-1} \left(1-\frac{i3^d}{n}\right) \notag \\
								&\geq& \left( 1-\frac{2k3^d}{n}\right)^{2k} \notag \\
\Rightarrow P[R_{i}\text{ is d-good for }\tset] &\leq& 1- \left( 1-\frac{2k3^d}{n}\right)^{2k} \notag \\
									&=& 1- \exp\left(2k \ln \left( 1-\frac{2k3^d}{n}\right)\right) \notag \\
									&\leq& 1 - \exp\left(-8k^{2}3^{d}/n\right) \notag \\
									&\leq& 8k^{2}3^{d}/n, \notag
\end{eqnarray}
where the last two inequalities follow from $\ln (1-x) > -2x$ for $x \le 1/2$ (note that $8 k^2 3^d/n < 1$ implies 
$2k 3^d/n \le 1/2$) and $e^x > 1+x$, respectively.
\end{proof}

Observe that if $R_{i}$ is not $d$-good, then the $k$-finger bound of the access sequence $[S_{i}]^{X/2k}$ is $\Omega(d(X-k))=\Omega(dX)$. This is because in every occurrence of $S_{i}$, there will be some $k$ elements out of the $2k$ total that will be outside the $d$-radius balls centered at the current $k$ fingers.

We call the entire sequence $S$ $d$-good for $\tset$ if at least half of the sets $R_{i}$ are $d$-good for $\tset$. Thus if $S$ is not $d$-good, then $\F^{k}_{\tset}(S) = \Omega(XYd)$.

\begin{lemma}
$P[S\text{ is d-good for }\tset] \leq \left(\frac{32k^{2}3^{d}}{n} \right)^{Y/2}$. 
\end{lemma}

\begin{proof}
By the previous lemma and by definition of goodness of $S$, we have that

\begin{align*}
P[S\text{ is d-good for }\tset] \leq {Y \choose Y/2}\left(\frac{8k^{2}3^{d}}{n} \right)^{Y/2} 
			 \leq 4^{Y/2}\left(\frac{8k^{2}3^{d}}{n} \right)^{Y/2} 
			= \left(\frac{32k^{2}3^{d}}{n} \right)^{Y/2}.
\end{align*}
\end{proof}

\noindent The theorem now follows easily. Taking a union bound over all BSTs on $[n]$, we have 
\[ P[S\text{ is d-good for some BST }\tset] \leq 4^{n}\left(\frac{32k^{2}3^{d}}{n} \right)^{Y/2}.
\]
Now set $Y=2n$. We have that 
\[ P[\exists \text{ a BST }\tset:\F^{k}_{\tset}(S) \leq md/4] \leq 4^{n}\left(\frac{32k^{2}3^{d}}{n} \right)^{n}.
\]
Putting $d = \log_{3} \frac{n}{256k^{2}}$ gives that for some constant $c<1$,
\[ 
 P[\exists \text{ a BST }\tset:\F^{k}_{\tset}(S) \leq c(m \log (n/k))] \leq 4^{n}\left(\frac{32k^{2}3^{d}}{n} \right)^{n} = 1/2 \]
which implies that with probability at least $1/2$ one of the sequences in our random construction will have $k$-finger bound that is $\Omega(m \log (n/k))$. The working set bound is always $O(m \log k)$. This establishes the theorem.



\bibliography{article}


\end{document}